\documentclass[final,1p,times,numbers]{elsarticle}
\usepackage{amsmath}
\usepackage{amsthm}
\usepackage{amssymb}
\usepackage{amsfonts}

\makeatletter
  {\par\addvspace{\@bls \@plus 0.5\@bls \@minus 0.1\@bls}\noindent
   {\bfseries\Elproofname}\enspace\ignorespaces}%
  {\par\addvspace{\@bls \@plus 0.5\@bls \@minus 0.1\@bls}}
\makeatother

\usepackage[numbers]{natbib}
\usepackage{lmodern}
\usepackage{url}
\usepackage{mdwlist}
\usepackage{multicol}
\usepackage{multirow}
\usepackage{enumerate}
\usepackage{paralist}
\usepackage{float}
\floatstyle{plain}
\newfloat{protocol}{htbp}{lop}
\floatname{protocol}{Protocol}

\usepackage{dsfont}
\usepackage{rotating}
\usepackage{booktabs}
\usepackage{xspace}
\usepackage{mathtools, cuted}
\usepackage{arydshln} %
\usepackage{paralist} %
\usepackage{gnuplottex}

\usepackage{algorithm}
\usepackage[noend]{algorithmic}
\makeatletter
\newcommand{\IFTHEN}[3][default]{\ALC@it\algorithmicif\ #2\
  \algorithmicthen\ #3\
  \ifthenelse{\boolean{ALC@noend}}{}{\algorithmicendif\ } \ALC@com{#1}}
\newcommand{\IFTHENELSE}[3]{\ALC@it\algorithmicif\ #1\
  \algorithmicthen\ #2\
  \algorithmicelse\ #3\
  \ifthenelse{\boolean{ALC@noend}}{}{\algorithmicendif\ } }
\newcommand{\IFTHENEND}[2]{\ALC@it\algorithmicif\ #1\
  \algorithmicthen\ #2\ \algorithmicendif}
\makeatother

\usepackage{color}


\usepackage{todonotes}

\newcommand{\supbr}[1]{{}^{[#1]}}
\newcommand{\supabr}[1]{{}^{\langle#1\rangle}}
\newcommand{\swplus}{\text{$+$}} %
\newlength{\colsp}

\newlength{\spaceafterprotocol}
\DeclareMathOperator{\rank}{rank}

\newtheorem{theorem}{Theorem}

\newtheorem{corollary}{Corollary}
\newtheorem{lemma}{Lemma}
\newtheorem{proposition}{Proposition}
\newcommand{\F}{\ensuremath{\mathbb{F}}}
\newcommand{\FF}{\ensuremath{\mathbb{F}}}%
\newcommand{\Z}{\ensuremath{\mathbb{Z}}}
\newcommand{\Q}{\ensuremath{\mathbb{Q}}}
\newcommand{\Primes}{\ensuremath{\mathbb{P}}}
\newcommand{\RPM}[1]{\ensuremath{\mathcal{R}_{#1}}\xspace}
\newcommand{\leftinv}[1]{\ensuremath{{#1}^{-1}_{\text{left}}}\xspace}
\newcommand{\SO}[1]{\ensuremath{\widetilde{O}({#1})}\xspace}
\newcommand{\GO}[1]{\ensuremath{O ({#1})}\xspace}
\newcommand{\sample}[1]{\ensuremath{\xhookleftarrow{\text{u.i.d.}} {#1}}\xspace}

\usepackage{mathtools}

\usepackage{letltxmacro}
\LetLtxMacro\orgvdots\vdots
\LetLtxMacro\orgddots\ddots

\makeatletter
\DeclareRobustCommand\vdots{%
  \mathpalette\@vdots{}%
}
\newcommand*{\@vdots}[2]{%
  \sbox0{$#1\cdotp\cdotp\cdotp\m@th$}%
  \sbox2{$#1.\m@th$}%
  \vbox{%
    \dimen@=\wd0 %
    \advance\dimen@ -3\ht2 %
    \kern.5\dimen@
    \dimen@=\wd2 %
    \advance\dimen@ -\ht2 %
    \dimen2=\wd0 %
    \advance\dimen2 -\dimen@
    \vbox to \dimen2{%
      \offinterlineskip
      \copy2 \vfill\copy2 \vfill\copy2 %
    }%
  }%
}
\DeclareRobustCommand\ddots{%
  \mathinner{%
    \mathpalette\@ddots{}%
    \mkern\thinmuskip
  }%
}
\newcommand*{\@ddots}[2]{%
  \sbox0{$#1\cdotp\cdotp\cdotp\m@th$}%
  \sbox2{$#1.\m@th$}%
  \vbox{%
    \dimen@=\wd0 %
    \advance\dimen@ -3\ht2 %
    \kern.5\dimen@
    \dimen@=\wd2 %
    \advance\dimen@ -\ht2 %
    \dimen2=\wd0 %
    \advance\dimen2 -\dimen@
    \vbox to \dimen2{%
      \offinterlineskip
      \hbox{$#1\mathpunct{.}\m@th$}%
      \vfill
      \hbox{$#1\mathpunct{\kern\wd2}\mathpunct{.}\m@th$}%
      \vfill
      \hbox{$#1\mathpunct{\kern\wd2}\mathpunct{\kern\wd2}\mathpunct{.}\m@th$}%
    }%
  }%
}

\newenvironment{smatrix}{\left[\begin{smallmatrix}}{\end{smallmatrix}\right]}
\newcommand{\checks}[1]{\ensuremath{\mathrel{\stackrel{?}{#1}}}}

\newlength{\arrowlength}
\newcommand{\fxrightarrow}[1]{\xrightarrow{\mathmakebox[\arrowlength]{#1}}}
\newcommand{\fxleftarrow}[1]{\xleftarrow{\mathmakebox[\arrowlength]{#1}}}

\newcommand{\LINSYS}{\text{\itshape LINSYS}} %

\usepackage{color,svgcolor}
\usepackage[plainpages=true]{hyperref}
\makeatletter
\hypersetup{
  pdftitle='RPM certificates',
  pdfauthor='JGD-EK-DL-CP',
  breaklinks=true,
  colorlinks=true,
  linkcolor=darkred,
  citecolor=blue,
  urlcolor=darkgreen,
}
\makeatother
\newcommand*{\doi}[1]{\href{http://dx.doi.org/#1}{\texttt{doi: \nolinkurl{#1}}}}

\newcommand{\OpenDreamKit}{the \href{http://opendreamkit.org}{OpenDreamKit} \href{https://ec.europa.eu/programmes/horizon2020/}{Horizon 2020} \href{https://ec.europa.eu/programmes/horizon2020/en/h2020-section/european-research-infrastructures-including-e-infrastructures}{European Research Infrastructures} project (\#\href{http://cordis.europa.eu/project/rcn/198334_en.html}{676541})}

\newcommand{\GACI}{the French National Research Agency program
  (\href{http://www.agence-nationale-recherche.fr/ProjetIA-15-IDEX-0002}{ANR-15-IDEX-02})}

\begin{document}

\begin{frontmatter}

\title{Elimination-based certificates for triangular equivalence and rank
  profiles\footnote{This work is partly funded by
    \OpenDreamKit~and \GACI~(Dumas, Lucas, Pernet) and by NSF (USA), grants
    CCF-1421128 and CCF-1717100 (Kaltofen).}}

\author{Jean-Guillaume Dumas}
\address{Univ. Grenoble Alpes, CNRS, Grenoble INP\footnote{Institute of Engineering Univ. Grenoble Alpes}, LJK, 38000 Grenoble, France}
\ead{Jean-Guillaume.Dumas@univ-grenoble-alpes.fr}
\ead[url]{http://www-ljk.imag.fr/membres/Jean-Guillaume.Dumas/}

\author{Erich Kaltofen}
\address{North Carolina State University, Department of Mathematics, \\Raleigh, North Carolina 27695-8205, USA}
\ead{kaltofen@math.ncsu.edu}
\ead[url]{http://www.kaltofen.us}

\author{David Lucas}
\address{Univ. Grenoble Alpes, CNRS, Grenoble INP\footnotemark[2], LJK, 38000 Grenoble, France}
\ead{David.Lucas@univ-grenoble-alpes.fr}
\ead[url]{http://www-ljk.imag.fr/membres/David.Lucas/}

\author{Cl\'ement Pernet}
\address{Univ. Grenoble Alpes, CNRS, Grenoble INP\footnotemark[2], LJK, 38000 Grenoble, France}
\ead{Clement.Pernet@univ-grenoble-alpes.fr}
\ead[url]{http://www-ljk.imag.fr/membres/Clement.Pernet/}

\begin{abstract}
In this paper, we give novel certificates for triangular equivalence
and rank profiles. 
These certificates enable somebody to verify the row or column rank profiles or
the whole rank profile matrix faster than
recomputing them, with a negligible overall overhead.
We first provide quadratic time and space non-interactive certificates
saving the logarithmic factors of previously known ones.
Then we propose interactive certificates for the same problems
whose Monte Carlo verification complexity requires a small constant
number of matrix-vector
multiplications, a linear space, and a linear number of extra field operations,
with a linear number of interactions.
As an application we also give an interactive protocol, certifying the
determinant or the signature of dense matrices, faster for the Prover than the
best previously known one. 
Finally we give linear space and constant round certificates for the row or
column rank profiles.
\end{abstract}
\end{frontmatter}
\tableofcontents

\newpage

\section{Introduction}
Within  the setting of verifiable computing, we propose in this paper
{\em interactive certificates} with the taxonomy
of~\cite{dk14}.
Indeed, we consider a protocol where a {\em Prover} performs a
computation and provides additional data structures or exchanges with a
{\em Verifier} who will use these to 
check the validity of the result, faster than by just recomputing it.
More precisely, in an interactive certificate,
the Prover submits
a {\em Commitment}, that is some result of a computation;
the Verifier answers by
a {\em Challenge}, usually some uniformly sampled random values;
the Prover then answers with
a {\em Response}, that the Verifier can use to convince himself of the validity
of the commitment. 
Several {\em rounds} of challenge/response might be necessary
for the Verifier to be fully convinced.

By Prover (resp. Verifier) {\em time}, we thus mean bounds on the number of arithmetic
operations performed by the Prover (resp. Verifier) during the
protocol, while by extra {\em space}, we mean bounds on the volume of
data being exchanged, not counting the size of the input and output of the computation. 

Such protocols are said to be {\em complete} if the probability
that a true statement is rejected by the Verifier can be made arbitrarily
small; and {\em sound} if the probability that a false
statement is accepted by the Verifier can be made arbitrarily small.
In practice it is sufficient that those probabilities are $<1$, as the protocols
can always be run several times.
Some certificates will also be {\em perfectly complete}, that
is a true statement is never rejected by the Verifier. 
All these certificates can be simulated non-interactively by the Fiat-Shamir
heuristic~\cite{Fiat:1986:Shamir}: publicly and uniformly sampled random values
produced by the Verifier are replaced by cryptographic hashes of the
input and of previous messages in the protocol. Complexities are
preserved.

Our protocols follow the proof-of-work protocols
of \cite{Goldwasser:2008:delegating,GKR15}
in that they verify that the Prover has performed
some LU matrix factorization.
However, they do so by verifying the factorization and the triangularity
of the factors, which remain stored on the Prover side
and are not communicated to the Verifier,
rather than verifying the entire circuit that computes
those factors by Lund-Fortnow-Karloff-Nisan polylog-compressive
sumcheck protocols.  In
\cite{jgd:2017:gkr} %
we have
applied \cite{GKR15} to matrices of exponential dimensions
where the entries are computed from their indices by
efficient circuits.
Our version of the GKR proof-of-work protocol
has a Verifier complexity that is, within a polylog
factor, the depth of a parallel circuit whose local
structure can be compute in polylog time, plus one
linear scan of the input.
The Prover complexity is within a polylog factor of
the size of the circuit.
The protocols here avoid those polylog factors.

It is possible to reduce the communication complexity
in~\cite{GKR15} to a constant number of rounds by
when the space complexity is bounded %
\cite{RRR16}
but it is not apparent to us how to asymptotically preserve the Prover's
time complexity then %
(it remains polynomial-time).
We will consider an $m\times n$  matrix $A$ of rank $r$ over a field $\F$.
The \emph{row rank profile}
of $A$ is the lexicographically minimal sequence of $r$ indices of independent
rows
of $A$.
Matrix $A$ has \emph{generic row
  rank profile} if its row
rank profile is $(1, \dots, r)$.
The \emph{column rank profile} is defined similarly on the columns of $A$.
Matrix $A$ has generic rank profile if its $r$ first leading principal minors
are nonzero. 
The \textit{rank profile matrix} of $A$, denoted by $\mathcal{R}_A$ is the
unique $m \times n$ $\{0,1\}$-matrix with $r$ nonzero entries, of which every
leading sub-matrix has the same rank 
as the corresponding sub-matrix of $A$. It is possible to compute
$\mathcal{R}_A$ with a deterministic algorithm in
$\GO{mnr^{\omega-2}}$ or with a Monte-Carlo probabilistic algorithm in $(r^{\omega} + m + n + \mu(A))^{1 + o(1)}$ field operations \cite{jgd:2017:bruhat},
where $\mu(A)$ is the
worst case %
arithmetic cost to multiply $A$ by a vector.

We first propose quadratic, space and verification time,
  non-in\-ter\-ac\-tive practical certificates for the row or column rank profile and
  for the rank profile matrix that are rank-sensitive. Previously
  known certificates have additional logarithmic factors to the
  qua\-dra\-tic complexities: replacing matrix multiplications by
  qua\-dra\-tic verifications in recursive algorithms yields at least one $\log(n)$
  factor~\cite{kns11}, graph-based approaches cumulate this and other
  logarithmic factors, at least from a compression by magical graphs
  and from a dichotomic search~\cite{sy15}.  

  We then propose two linear space interactive
  certificates. The first certificate is used to prove that two non-singular
  matrices are triangular equivalent (i.e. there is a triangular change of basis
  from one to the other). The second certificate is used to prove that a matrix
  has a generic rank profile.
  These two certificates are then applied to certify the row or column rank
  profile, the $P$ (permutation) and $D$ (diagonal) factors of a LDUP
  factorization, the determinant and the rank profile matrix.
  These certificates require, for the Verifier, between 1 and 4 applications of
  $A$ to a vector and a linear number of field operations. They are still
  elimination-based for the Prover, but do not require to 
  communicate the obtained triangular decomposition. 

  An interesting setting would be for instance the case when the matrix $A$ is
  sparse. Blackbox methods could then be used, when elimination-based method
  would suffer from some fill-in. Quite often though, elimination-based methods
  are then more limited by the available memory than by the number of
  computation. A Verifier could then outsource its computations to a server,
  for which fill-in would not be an issue, and use only still sparse
  matrix-vector multiplications to Verify the result.

For instance, for the Determinant, our new certificates require the computation
of a PLUQ decomposition for the Prover, linear communication and Verifier time,
with no restriction on the field size.
The previously best known certificate for the determinant required instead some
characteristic polynomial ({\sc{CharPoly}}) computations.

With respect to~\cite{jgd:2017:rpmcert} we propose a complete analysis 
of the rank profile matrix certificate~\ref{cert:RPM:full} only sketched there; 
an application to computing the signature of a symmetric integral matrix; and a
whole set of new certificates: for triangular equivalence, row and
column rank profile, we are now able to propose protocols that preserve Prover
and Verifier efficiency, while reducing the number of rounds from linear to
constant.
The constant round complexity is an important additional bonus in the
delegation scenario, where network latency can make communication rounds
more expensive. 
Note that the probabilistic analysis of \cite[Theorem~4]{jgd:2017:rpmcert}
omitted to account for several possibilities of failure,
which is corrected here yielding a smaller probability of
detecting a dishonest Prover.
We identify the symmetric group with the group of permutation matrices, and
write $P\in \mathcal{S}_n$ to denote that a matrix $P$ is a permutation
matrix. There, $P[i]$ is the row index of the nonzero element 
of its $i$-th column; $\mathcal{D}_n(\F)$ is the group of
invertible diagonal matrices over the field $\F$;
$\mathcal{D}^{(2)}_n(\F)$ represents block diagonal matrices with diagonal or
anti-diagonal blocks of size $1$ or $2$.
For two subsets of row indices $\mathcal{I}$ and of column indices
$\mathcal{J}$, $A_{\mathcal{I},\mathcal{J}}$ denotes the submatrix extracted
from $A$ in these rows and columns.
The set of prime numbers will be denoted by $\Primes$.
Lastly, $x \sample{S}$ denotes that $x$  is
uniformly independently randomly sampled %
from~$S$. In what follows, while computing the communication space, we consider
that field elements and indices have the same size.
\section{Non interactive and quadratic communication certificates}
\label{sec:noninterractive}
In this section, we propose two certificates, first for the column (resp. row)
rank profile, and, second, for the rank profile matrix.
While the certificates have a quadratic space communication complexity, they have the advantage
of being non-interactive.

\subsection{Freivalds' certificate for matrix product}
\label{subsec:freivalds}

In this paper, we will use Freivalds' certificate \cite{freivalds79} to verify matrix multiplication.
Considering three matrices $A, B$ and $C$ in $\mathbb{F}^{n \times n}$, such that 
$A \times B = C$, a straightforward way of verifying the equality would be to perform the multiplication
$A \times B$ and to compare its result coefficient by coefficient with $C$. While this method is
deterministic, it has a time complexity of $O(n^\omega)$, which is the matrix multiplication complexity.
As such, it cannot be a certificate, as there is no complexity difference between the computation and
the verification.

\begin{protocol}[htbp]
    \centering
    \begin{tabular}{|l c l|}
        \hline
        Prover & & Verifier \\
        & $A, B\in \mathbb{F}^{n \times n}$
        & \\
        \hdashline\rule{0pt}{12pt}
        $C= A  B $ & $\fxrightarrow{\text{C}}$ &
        Choose ${S}\subset{\F}$\\
        & & $v\sample{S}^{n}$\\
        & & $ A  (B  v)- C  v \stackrel{?}{=} 0 $\\\hline
\end{tabular}
    \caption{Freivalds' certificate for matrix product}
    \label{cert:freivalds}
    \vspace{\spaceafterprotocol}
\end{protocol}

Freivalds' certificate proposes a probabilistic method to check this product in a time complexity of
$\mu(A)+\mu(B)+\mu(C)$ using matrix/vector multiplication, as detailed in
Protocol~\ref{cert:freivalds}.

\subsection{Column rank profile certificate}
\label{sec:noninterractive:CRP}

We now propose a certificate for the column rank profile.

\begin{protocol}[htbp]
  \centering
  \begin{tabular}{|l c r|}
    \hline
    \multicolumn{1}{|c}{Prover} & & \multicolumn{1}{c|}{Verifier} \\
    & $A \in \mathbb{F}^{m \times n}$ & \\
    \hdashline
    \multirow{2}{140pt}{a $PLUQ$ decomposition of $A$ s.t. $UQ$ is in
      row echelon form} 
    &\multirow{2}{*}{$\fxrightarrow{\text{P,L,U,Q}}$} 
    &\rule{0pt}{10pt}{$UQ$ row echelonized?}\\
    \rule{0pt}{30pt}&  & $A\stackrel{?}{=}PLUQ$, by Protocol~\ref{cert:freivalds}\\
    \hdashline\rule{0pt}{10pt}
    & & Extract $Q[1],\ldots, Q[r]$\\
    \hline
  \end{tabular}
  \caption{Column rank profile, non-interactive}
  \label{cert:crp:noninter}
\vspace{\spaceafterprotocol}\end{protocol}

\begin{lemma}
    \label{lem:crp:ni}
    Let $A=PLUQ$ be the PLUQ decomposition of an $m \times n$ matrix $A$ of rank
    $r$. If  $UQ$ is in row echelon form then $(Q[1], \dots, Q[r])$ is the column rank profile of $A$.
\end{lemma}

\begin{proof}
  Write $A=P\begin{smatrix} L_1 \\ L_2 \end{smatrix}\begin{smatrix} U_1 & U_2 \end{smatrix}Q$, where $L_1$ and
  $U_1$ are $r\times r$ lower and upper triangular respectively.
  If $UQ$ is in echelon form, then  
  $ R  = \begin{smatrix} \begin{smallmatrix}I_r & U_{1}^{-1} U_2\end{smallmatrix} \\ {0_{(m-r)\times n}} \end{smatrix}$
  is in reduced echelon form.
  Now
  $$
  \begin{bmatrix} U_1^{-1}\\&I_{m-r} \end{bmatrix}
  \begin{bmatrix} L_1\\L_2&I_{m-r}  \end{bmatrix}^{-1} P^TA= 
  \begin{bmatrix} U_1^{-1}UQ\\0_{(m-r)\times n} \end{bmatrix}=
  R$$
  is left equivalent to $A$ and is therefore the echelon form of $A$.
  Hence the sequence of column positions of the pivots in $R$, that is
  $(Q[1],\dots,Q[r])$, is the 
  column rank profile of $A$.
\end{proof}

Lemma~\ref{lem:crp:ni} provides a criterion to verify a column rank
profile from a PLUQ decomposition.
Such decompositions can be computed in practice by several variants of Gaussian
elimination, with no arithmetic overhead, as shown in~\cite{jps13} or~\cite[\S~6.4]{jgd:2017:bruhat}.
Hence, we propose the certificate in Protocol~\ref{cert:crp:noninter}.

\begin{theorem}
    Let $A \in \mathbb{F}^{m \times n}$ with $r=\rank(A)$.
    Certificate~\ref{cert:crp:noninter},  verifying the column rank profile of $A$
    is sound, perfectly complete, with a communication bounded by $O(r(m+n))$, a Prover
    computation cost bounded by $O(mnr^{\omega -2})$ and
    a Verifier computation cost bounded by $O(r(m+n)) + \mu(A)$.
\end{theorem}

\begin{proof}
    If the Prover is honest, then, $UQ$ will be in row echelon form and $A=PLUQ$, thus, 
    by Lemma~\ref{lem:crp:ni}, the Verifier will be able to read the column rank
    profile of~$A$ from~$Q$.
    If the Prover is dishonest, either $A \neq PLUQ$, which will be caught by the Prover with probability 
    $p \geq 1-\frac{1}{|S|}$ using Freivalds' certificate \cite{freivalds79} or $UQ$ is not in row echelon from, which
    will be caught every time by the Verifier.

    The Prover sends $P, L, U \text{ and } Q$ to the Verifier, hence the communication cost of $O(r(m+n))$, as
    $P$ and $Q$ are permutation matrices and $L, U$, are respectively $m \times r$ and $r \times n$ matrices,
    with $r = rank(A)$.
    Using algorithms provided in \cite{jps13}, one can compute the expected $PLUQ$ decomposition in
    $O(mnr^{\omega -2})$.
    The Verifier has to check if $A = PLUQ$, and if $UQ$ is in row echelon form, which can be done in $O(r(m+n))$.
\end{proof}

Note that this holds for the row rank profile of $A$: in that case, the Verifier has to check if $PL$ is in
column echelon form.

\subsection{Rank profile matrix certificate}
\label{sec:noninterractive:RPM}

\begin{lemma}\label{lem:echelonized}
  A decomposition  $A=PLUQ$
  reveals the rank profile matrix, namely
  $\RPM{A}=P \begin{smatrix}  I_r\\&0  \end{smatrix}Q$, if and only if
  $P \begin{smatrix} L&0  \end{smatrix}P^T$ is lower triangular and
  $Q^T \begin{smatrix} U\\0  \end{smatrix} Q$ is upper triangular.
\end{lemma}

\begin{proof}

  The \textit{only if} case is proven in~\cite[Th.~21]{jgd:2017:bruhat}.
Now suppose that 
$P\begin{smatrix}L&0_{m\times(m-r)} \end{smatrix}P^T$ is lower triangular.
Then we must also have that
$\overline{L}=P\begin{smatrix}L&
\begin{smallmatrix}  0\\I_{m-r}\end{smallmatrix}
 \end{smatrix}P^T$
is lower triangular and non-singular.  Similarly suppose that
$Q^T \begin{smatrix} U\\0  \end{smatrix} Q$ is upper triangular so that
$\overline{U}=Q^T\begin{smatrix}U\\
\begin{smallmatrix}  0&I_{n-r}\end{smallmatrix}\end{smatrix}
Q$
is non-singular upper triangular.
We have $A=\overline{L}P \begin{smatrix}  I_r\\&0\end{smatrix}Q\overline{U}$.
Hence the rank of any $(i,j)$ leading submatrix of $A$ is that of the $(i,j)$
leading submatrix of $P\begin{smatrix}  I_r\\&0\end{smatrix}Q$, thus proving
that $\RPM{A}=P\begin{smatrix}  I_r\\&0\end{smatrix}Q$.
\end{proof}

We use this characterization to verify the computation of the rank
profile matrix in the following protocol:
Once the Verifier receives $P, L, U \text{ and } Q$, he has to
check that $A = PLUQ$, using Freivalds' certificate \cite{freivalds79}, 
and check that $L$ is echelonized by $P$ and $U^T$ by~$Q^T$.
If successful, the Verifier can just compute the rank profile matrix of $A$ from
$P$ and $Q$, as shown in Protocol~\ref{cert:RPM:noninter}.

\begin{protocol}[htbp]
    \centering
    \begin{tabular}{|c c p{6cm}|}
      \hline
       Prover & & \hspace{15pt}Verifier \\
        \multicolumn{3}{|c|}{$A \in \mathbb{F}^{m \times n}$} \\
        \hdashline\rule{0pt}{12pt}
         \multirow{3}{3.4cm}{a PLUQ decomp. of $A$ revealing  $\RPM{A}$.} &
         $\xrightarrow{\text{P,L,U,Q}}$ &
         1.  $A\stackrel{?}{=}PLUQ$ by
         Protocol~\ref{cert:freivalds}\\
         & & 2. Is $PLP^T$ lower triangular?\\
         & & 3. Is $Q^TUQ$ upper triangular?\\
         & & Extract $\mathcal{R}_A = P \begin{bsmallmatrix} I_{r} & \\ &
           0_{(m-r) \times (n-r)} \end{bsmallmatrix} Q$\\
         \hline
    \end{tabular}
    \caption{Rank profile matrix, non-interactive}
    \label{cert:RPM:noninter}
\vspace{\spaceafterprotocol}\end{protocol}

\begin{theorem}
    Certificate~\ref{cert:RPM:noninter} verifies the rank profile matrix of $A$,
    it is sound and perfectly complete, with a communication cost bounded by
    $O(r(m+n))$, a Prover computation cost bounded by $O(mnr^{\omega-2})$ and a
    Verifier computation cost bounded by $O(r(m+n)) + \mu(A)$. 
\end{theorem}

\begin{proof}
    If the Prover is honest, then, the provided $PLUQ$ decomposition is indeed
    a factorization of $A$, which means Freivalds' certificate will pass.
    It also means this $PLUQ$ decomposition reveals the rank profile matrix.
    According to Lemma~\ref{lem:echelonized}, $PLP^T$ will be lower triangular
    and $Q^TUQ$ upper triangular. Hence the verification will succeeds and 
    $\mathcal{R}_A = P \begin{bsmallmatrix} I_{r} & \\ & 0_{(m-r) \times (n-r)} \end{bsmallmatrix} Q$
    is indeed the rank profile matrix of $A$.
    If the Prover is dishonest, either $A \neq PLUQ$, which will be caught with probabilty
    $p \geq 1-\frac{1}{|S|}$ by Freivalds' certificate or the $PLUQ$
    decomposition does not reveal the rank profile matrix of $A$.
    In that case, Lemma~\ref{lem:echelonized} implies that either
    $P \begin{smatrix}      L&0    \end{smatrix}P^T$ is not lower triangular or
    $P \begin{smatrix}      U\\ 0 \end{smatrix}Q$ is not upper triangular which
    will be detected.

    The Prover sends $P, L, U \text{ and } Q$ to the Verifier, hence the
    communication cost of $\GO{r(m+n)}$.
    A rank profile matrix revealing $PLUQ$ decomposition can be computed in $O(mnr^{\omega-2})$ operations~\cite{DPS:2013}.
    The Verifier has to check if $A = PLUQ$, which can be achieved in
    $O(r(m+n))+\mu(A)$ field operations.
\end{proof}

\section{Linear communication certificate toolbox}

\subsection{Triangular one sided equivalence}

Two matrices $A, B \in \mathbb{F}^{m \times n}$ are right (resp. left)
equivalent if there exist an invertible $n\times n$ matrix $T$ such that $AT=B$
(resp. $TA=B$). If in addition $T$ is a lower triangular matrix, we say that $A$ and
$B$ are lower triangular right (resp. left)  equivalent. The upper triangular
right (resp. left ) equivalence is defined similarly.
We propose a certification protocol that two matrices are left or right
triangular equivalent.
Here, $A$ and $B$ are input, known by the Verifier and the Prover, and $A$ is
supposed to be \emph{regular} (full rank).
A simple certificate would be the matrix $T$ itself, in which case the Verifier would
check the product $AT = B$ using Freivalds' certificate.
This certificate is non-interactive and requires a quadratic number of communication.
In what follows, we present a certificate which allows to verify the one sided
triangular equivalence without communicating $T$, requiring only $2n$
communications.
It is essentially a Freivalds' certificate with a constrained interaction
pattern in the way the challenge vector and the response vector are
communicated.
This pattern imposes a triangular structure in the way the Provers' responses
depend on the Verifier challenges.

\begin{protocol}[htbp]
    \centering
    \begin{tabular}{|p{4.4cm} c l|}
        \hline
        Prover &  & Verifier \\
        & $ A, B \in \mathbb{F}^{m{\times}n} $ &  \\
        & $A$ regular, $m{\geq}n$ &  \\
        \hdashline\rule{0pt}{12pt}
        $T$ lower triangular matrix s.t. $AT = B$ & $\xrightarrow{\text{1: T exists}}$ & \\
        \hdashline\rule{0pt}{12pt}
        \multirow{2}{*}{$y_1 = T_{1,*} \begin{smatrix}x_1\\0\\ \vdots \end{smatrix} $}& $\xleftarrow{\text{$2: x_1$}}$ & $x_i\sample{S}\subset{\F}$ \\
         & $\xrightarrow{\text{$3: y_1$}}$ & \\
        \multicolumn{1}{|c}{$\vdots$}&$\vdots$& \\
        \multirow{2}{*}{$y_n = T_{n,*} \begin{smatrix}x_1\\ \vdots\\x_n \end{smatrix} $}& $\xleftarrow{\text{$2n: x_n$}}$ &  \\
        & $\xrightarrow{\text{$2n+1: y_n$}}$ &
        $y =\begin{bmatrix}y_1&..&y_n\end{bmatrix}^T$ \\
        & &
        $Ay\stackrel{?}{=}Bx$\\
        \hline
    \end{tabular}
    \caption{Lower triang. right equivalence of regular matrices}
    \label{cert:triangular_equivalence}
\vspace{\spaceafterprotocol}\end{protocol}

\begin{theorem}
    \label{th:triangular_equivalence}
    Let $A, B \in \mathbb{F}^{m \times n}$, $m{\geq}n$, and assume $A$ is regular.
    Certificate~\ref{cert:triangular_equivalence}
    proves that there exists a lower triangular matrix $T$ such that $AT = B$.
    This certificate is sound, with probabilty larger than $1-\frac{1}{|S|}$,
    perfectly complete and occupies $2n$
    communication space. The Prover complexity is $O(mn^{\omega-1})$ field operations and  
    the Verifier computation cost  is $\mu(A) + \mu(B)$ field operations.
\end{theorem}

\begin{proof}
  If the Prover is honest, then $AT=B$ with $T$ triangular and she just computes
  $y=Tx$, so that $Ay=ATx=Bx$.
  If the Prover is dishonest, then she %
must try to convince the Verifier even if
  the matrices are inequivalent.
  For the sake of the argument, replace the random values $x_1,\dots,x_n$ by
  algebraically independent variables $X_1,\dots,X_n$. 
  Then there are two cases, either $AT\neq B$ for any $T$
  or there exists at least one such matrix $T$ but none of them are triangular.

  In the former case, $AT\neq B$, there is thus at least one inconsistent column
  in $B$, say the $j$-th. Then, there exists a Farkas' 
 certificate of
  inconsistency for that column (a vector $z$ such that $z^TA=0$ and
  $z^TB_{*,j}\neq 0$). This means that $z^TAy=0$ for any $y$, but $z^T B [X_1,\dots,X_n]^T$
  is a not identically zero polynomial  (at least the coefficient of $X_j$
  is non zero) of degree $1$. Therefore, 
  by the DeMillo-Lipton/\allowbreak Schwartz/\allowbreak Zippel
  lemma~\cite{DeMillo:1978:ipl,Zippel:1979:ZSlemma,Schwartz:1979:SZlemma},
  its evaluation will be zero with probability at most $1/|S|$.

  In the latter case, $AT=B$ but $T$ is not triangular.
  Since $A$ is regular, there is thus a unique $n\times n$ matrix
  $T$ (that is, $T=\leftinv{A}B$, for any $\leftinv{A}$ left inverse of $A$)
  such that $AT=B$: indeed $T=\leftinv{A}AT=\leftinv{A}B$.
  For the same reason, the equality $Ay = Bx = ATx$ implies $y = Tx$.
  If $T$ is not lower triangular, there is a row-index $i$ such
  that the entry $t_{i,j_m}\ne 0$ for some $j_m>i$.  The test $y = Tx$
  only succeeds if $y_i = \sum_{j=0}^n t_{i,j}x_j$.  Now the Prover
  selects $y_i$ before $x_{j_m}$ is revealed.  Therefore, with
  probability no more than $1/|S|$ the Verifier selects the field element
  $x_{j_m} = 1/t_{i,j_m} (y_i - \sum_{j\ne j_m} t_{i,j}x_j)$, and
  the test succeeds for false $T$.

  This certificate requires to transmit $x$ and $y$, which costs $2n$ in communication.
    The Verifier has to compute $Ay$ and $Bx$, whose computational cost is $\mu(A) + \mu(B)$.
    The Prover has to compute $T$, this can be done by a PLUQ
    elimination on $A$ followed by a triangular system solve, both in
    $O(mn^{\omega-1})$. Then  $y=Tx$ requires only $O(n^2)$ operations.
\end{proof}

    Note that the case where $T$ is upper triangular works similarly: the Verifier needs to transmit
    $x$ in reverse order, starting by $x_n$.

\subsection{Generic rank profile-ness}

The problem here is to verify whether a non-singular input matrix
$A\in\F^{m\times n}$ has generic rank profile (to test non-singularity, one can
apply beforehand the linear communication certificate in~\cite[Fig.~2]{dk14},
see also Protocol~\ref{cert:lower_rank} thereafter).
A matrix $A$ has generic rank profile if and only if it has an LU decomposition
$A=LU$, with $L$ non-singular lower triangular and $U$ non-singular upper
triangular.  The protocol picks random vectors $\phi,\psi,\lambda$ and asks the Prover to
provide the vectors $z^T=\lambda^TL$, $x=U\phi$, $y=U\psi$ on the fly, while receiving the
coefficients of the vectors $\phi,\psi,\lambda$ one at a time.
These vectors satisfy the fundamental equations  %
$z^Tx=\lambda^TA\phi$ and $z^Ty=\lambda^TA\psi$ that will be checked by the Verifier.

\begin{protocol}[htbp]
  \centering
  \begin{tabular}{|l c l|}
    \hline\rule{0pt}{12pt}   
    Prover & & Verifier \\
     & $A \in \mathbb{F}^{n{\times}n}$ & \\
     & non-singular & \\
    \hdashline\rule{0pt}{12pt}    
     $A=LU $ &$\xrightarrow{\text{$A$ %
     has g.r.p.}}$ &\\
    \hdashline\rule{0pt}{12pt}   
    &\multicolumn{2}{l|}{\textbf{for} $i$ from $n$ downto $1$} \\
    $\begin{bmatrix} x&y \end{bmatrix} = U \begin{bmatrix}\phi&\psi \end{bmatrix}$ & $\xleftarrow{\phi_i,\psi_i}$ & \hspace{-2pt}$(\phi_i,\psi_i) \sample{S^2}\subset{\F^2}$\\
    & $\xrightarrow{x_i,y_i} $ &\\
    $z^T = \lambda^T L$& $\xleftarrow{\lambda_i}$ & \hspace{-2pt}$\lambda_i\sample{S}\subset{\F}$\\
    & $\xrightarrow{z_i} $ &\\
    \hdashline\rule{0pt}{12pt}    
    && \hspace{-2pt}$z^T  \begin{bmatrix}x&y\end{bmatrix}  \checks{=} (\lambda^T A)\begin{bmatrix}\phi&\psi\end{bmatrix}$ \\
    \hline
  \end{tabular}
  \caption{Generic rank profile with linear communication}
  \label{cert:GRP}
      \vspace{\spaceafterprotocol}
\end{protocol}

\begin{theorem}
\label{th:GRP}
   Certificate~\ref{cert:GRP} verifying that a non-singular matrix has generic
   rank profile is sound, with probability $\ge (1-\frac{1}{|S|})^{2n}$,
   perfectly complete, communicates $6n$ field elements, and can be computed in
   $O(n^\omega)$ field operations for the Prover and $\mu(A) + 8n$ field
   operations for the Verifier. 
\end{theorem}

\begin{proof}[Proof of Theorem~\ref{th:GRP}]
The protocol is perfectly complete: if $A=LU$, then
$z^T\begin{bmatrix}x&y\end{bmatrix}=
\lambda^TLU\begin{bmatrix}\phi&\psi\end{bmatrix}=\lambda^TA\begin{bmatrix}\phi&\psi\end{bmatrix}$,
and the answer of any honest Prover will pass the Verifier test.

For any $i$ such that the $(i-1)\times (i-1)$ leading submatrix of $A$ has
generic rank profile, we can write a partial LU decomposition of $A$ with the following notations:
\begin{equation}\label{eq:partialLU}
A = 
\underbrace{\begin{bmatrix}
L\supabr{i} & 0 \\
B\supabr{i} & I_{n-i+1}
\end{bmatrix}}_B
\;
\underbrace{\begin{bmatrix}
U\supabr{i} & V\supabr{i} \\
0 & C\supabr{i}
\end{bmatrix}}_C,
\end{equation}
where $L\supabr{i} \in \FF^{(i-1)\times (i-1)}$ is non-singular lower triangular,
$U\supabr{i} \in \FF^{(i-1)\times (i-1)}$ is non-singular upper triangular,
$B\supabr{i} \in \FF^{(n-i+1)\times (i-1)}$, $V\supabr{i}\in \FF^{(i-1)\times (n-i+1)}$,
\newline%
$C\supabr{i}\in\FF^{(n-i+1)\times(n-i+1)}$.

Let $v\supbr{i...n} = [v_i,\ldots,v_n]^T\in\FF^{n-i+1}$
for a vector $v\in\FF^n$, and let
\begin{equation}\label{eq:bilinC}
\eta_i =
(\lambda\supbr{i...n})^T C\supabr{i} \phi\supbr{i...n},\quad
\xi_i =
(\lambda\supbr{i...n})^T C\supabr{i} \psi\supbr{i...n}.
\end{equation}
Consider the following predicate: %
\begin{equation}\label{eq:predicate}
H_i\colon %
\eta_i {=} (z\supbr{i\dots n})^T x\supbr{i\dots n}
\text{ and }
\xi_i {=} (z\supbr{i\dots n})^T y\supbr{i\dots n}.
\end{equation}
Note that $H_1$ is
what the Verifier checks because then $B = I_n$.
Note also that when $A$ is in generic rank profile with $A = LU$
and $z^T = \lambda^T L$
and $x = U\phi$ and $y = U\psi$ then $H_i$ is true for all $i$.
To see this consider an LU-factorization
$C\supabr{i} = \bar L\supabr{i}\bar U\supabr{i}$
and the identity
\begin{equation}
A =
\begin{bmatrix}
L\supabr{i} & 0 \\
B\supabr{i} & I_{n-i+1}
\end{bmatrix}
\;
\begin{bmatrix}
U\supabr{i} & V\supabr{i} \\
0 & C\supabr{i}
\end{bmatrix}
=
\begin{bmatrix}
L\supabr{i} & 0 \\
B\supabr{i} & \bar L\supabr{i}
\end{bmatrix}
\;
\begin{bmatrix}
U\supabr{i} & V\supabr{i} \\
0 & \bar U\supabr{i}
\end{bmatrix} = LU.
\end{equation}
Then %
$(z\supbr{i...n})^T = (\lambda\supbr{i...n})^T \bar L\supabr{i}$
and $x\supbr{i...n} = \bar U\supabr{i} \phi\supbr{i...n}$ and $y\supbr{i...n} =
\bar U_\supabr{i} \psi\supbr{i...n}$ verify $H_i$.
Note that the conditions are only tested by the Verifier for $i=1$. %

At stage $i$, let $\Lambda_{i}, \Phi_i$ and $\Psi_i$ be variables for the
random choices for $\lambda_{i}, \phi_i$ and $\psi_i$ and $Z_i$ be a variable
for the Prover's choice of $z_i$. Then $H_i$ in (\ref{eq:predicate}) %
expands as: %

\begin{equation}\label{eq:efgh}
\left\{
\begin{array}{lll}
x_{i} Z_{i} &=& \Big(d %
\Phi_i+\overbrace{\sum_{j=i+1}^{n}
C\supabr{i}_{1,j-i+1} \phi_j}^e\Big) \Lambda_{i} + a\Phi_i+f,\\
y_{i} Z_{i} &=& \Big(d %
\Psi_i+\underbrace{\sum_{j=i+1}^{n}
C\supabr{i}_{1,j-i+1} \psi_j}_g\Big) \Lambda_{i} + a\Psi_i+h,
\end{array}
\right.
\end{equation}
where $d %
=C\supabr{i}_{1,1}$ and  $a=\sum_{k=i+1}^{n} \lambda_k C\supabr{i}_{k-i+1,1}$, or equivalently
\begin{equation}\label{eq:mat:i}
\begin{bmatrix}
-(d %
\Phi_i+e) & x_{i}
\\
-(d %
\Psi_i+g) & y_{i}
\end{bmatrix}
\begin{bmatrix}
\Lambda_{i} \\ Z_{i}
\end{bmatrix}
=
\begin{bmatrix}
a\Phi_i+f\\ a\Psi_i+h
\end{bmatrix}.
\end{equation}

Suppose now that $A$ is not in generic rank profile, and let $i_0$ be minimal
such that the leading $i_0\times i_0$ minor of $A$ is equal to~$0$. %
On any corresponding partial LU decomposition this means that $d %
=C\supabr{i_0}_{1,1} = 0$.
Furthermore, because $A$ is assumed to be non-singular,
there exist indices $k_0$ with $2 \le k_0 \le n-i_0+1$, and $j_0$ with $2\le j_0 \le n-i_0+1$
such that $C\supabr{i_0}_{k_0,1} \ne 0$ and $C\supabr{i_0}_{1,j_0} \ne 0$.

We will now prove the two following statements:
\begin{enumerate}
\item $H_{i_0}$ is false with probability $\ge %
(1-1/|S|)^4$; \label{induc:bc}
\item If $H_{i+1}$ is false then $H_{i}$ is false with probability $\ge %
(1-1/|S|)^2$ for $1\leq i< i_0$. \label{induc:heredity}
\end{enumerate}
Informally, this means that the Prover cannot %
achieve $H_{i_0}$ with any choice
of returned values $x_1,\ldots,z_{i_0}$, with high probability and then this
failure propagates with high probability to $H_1$ which is checked by the Verifier.
By induction, this leads to a probability of $\ge %
(1-1/|S|)^{4+2(i_0-1)}$ that the
Verifier check will fail when the matrix $A$ is not in generic rank
profile. Since $A$ is non-singular, $i_0\leq n-1$, and therefore this
probability is %
$\ge (1-1/|S|)^{2n}$.

First, we prove %
Statement~\ref{induc:bc},
that is the case when $d=0$. %
The Verifier selects a random
$\lambda_{i_0}$, and then the Prover a~$z_{i_0}$.
If the coefficient matrix in (\ref{eq:mat:i})
is non-singular, there is a unique solution for $\Lambda_{i_0}$, which the Verifier
will choose with probability $\le 1/|S|$.
Otherwise, the coefficient matrix is singular and the only way for the system to
have a solution is that the determinant
\begin{equation*}
\Delta=
\left|\begin{array}{cc}-e& a\Phi_{i_0}+f\\ -g & a\Psi_{i_0} +h \end{array} \right| =
-e (a\Psi_{i_0} +h)+ g (a\Phi_{i_0}+f)
\end{equation*}
is equal to~$0$, %
which exactly happens in the three following cases:
\begin{compactenum}[a.]
\item $\begin{bmatrix}e&g\end{bmatrix} = \begin{bmatrix}0&0\end{bmatrix}$, which
happens with probability $\leq 1/|S|^2 %
$ as $C\supabr{i_0}_{1,j_0} \ne 0$;
\item $a=0$ (which happens with probability $\leq 1/|S|$ as $C\supabr{i_0}_{k_0,1} \ne 0$) and
$\left|\begin{array}{cc}-e & f\\ -g & h \end{array}\right| = 0 %
$;
\item otherwise, $ea\neq 0$ or $ga\neq 0$ and $\Delta$ is a nonzero polynomial
of degree 1 in $\Phi_{i_0},\Psi_{i_0}$ and
evaluates to $0$ for the random choices $\phi_{i_0},\psi_{i_0}$ %
with probability $\leq 1/|S|$.
\end{compactenum}
Overall, $H_{i_0}$ is false with probability
\begin{equation*}
\ge \big(1 - \frac{1}{|S|^2}\big) \big(1-\frac{1}{|S|}\big)^3
\ge \big(1 - \frac{1}{|S|}\big)^4 \ge 1-\frac{4}{|S|}
\end{equation*}
based on the random choices of the Verifier:
$\phi_{j_0},\psi_{j_0}$ yield
$\begin{bmatrix}e & g\end{bmatrix}\neq \begin{bmatrix}0&0\end{bmatrix}$;
$\lambda_{k_0}$ yields $a \ne 0$;
$\phi_{i_0},\psi_{i_0}$ yield $\Delta\neq 0$;
$\lambda_{i_0}$ avoids the unique solution to (\ref{eq:mat:i}).

For Statement%
~\ref{induc:heredity}, consider 
the predicate $H_i$ (\ref{eq:predicate}) at $i<i_0$, that is $d\ne 0$. %
Similarly, if the coefficient matrix in (\ref{eq:mat:i})
is non-singular, there is a unique solution for $\Lambda_{i}$, which the Verifier
will choose with probability $\le 1/|S|$. Otherwise, the coefficient matrix is singular and the only way for the system to
have a solution is that the following determinant is equal~$0$:
\begin{equation*}
0 = %
\Delta=
\left|\begin{array}{cc}-(d %
\Phi_i+e)& a\Phi_{i}+f\\
-(d %
\Psi_i+g) & a\Psi_{i} +h \end{array} \right| =
 (d %
f-ae)\Psi_i -(d %
h-ag)\Phi_i -eh+ g f
.%
\end{equation*}
We block decompose the bottom right block in the incomplete right factor in (\ref{eq:partialLU}) %
$C\supabr{i}= \begin{bmatrix}d %
&r^T %
\\
s &W\end{bmatrix}$, where
$d %
=C\supabr{i}_{1,1}\neq 0$. %
We have $C\supabr{i+1}=W-\frac{1}{d}s r^T%
$.
Now since $a=(\lambda\supbr{i+1\dots n})^T%
s, e=r^T %
\phi\supbr{i+1\dots n}$, we have
$ae=(\lambda\supbr{i+1\dots n})^T%
sr^T%
\phi\supbr{i+1\dots n}$
and
\begin{eqnarray*}
f-\frac{ae}{d %
}&=&
(\lambda\supbr{i+1\dots n})^TC\supabr{i+1}\phi\supbr{i+1\dots n}-(z\supbr{i+1\dots n})^Tx\supbr{i+1\dots n}\\
&=&\eta_{i+1}-(z\supbr{i+1\dots n})^Tx\supbr{i+1\dots n}.
\end{eqnarray*}
Similarly, $h-\frac{ag}{d %
} = \xi_{i+1}-(z\supbr{i+1\dots n})^Tx\supbr{i+1\dots n}$,
and these two quantities are not equal to~$0$ %
simultaneously%
, for %
otherwise $H_{i+1}$ would be true.
Therefore $\Delta$ is a nonzero polynomial of degree $1$ in $\Phi$ and
$\Psi$. It is equal to~$0$ %
with probability $\leq 1/|S|$. Overall, $H_i$ is false with
probability $\ge %
(1-1/|S|)^2$ based on the random choices for $\lambda_i, \phi_i$
and $\psi_i$ made by the Verifier.

\newif\ifericharg\erichargfalse
\ifericharg
\input{ek_grp.tex}
\fi

Finally, for the complexity, the Prover needs one Gaussian elimination to
compute $LU$ in time $O(n^\omega)$, then her extra work is just three triangular
solve in $O(n^2)$. The extra communication is six vectors, $\phi,\psi,\lambda,
x, y, z$, and the
Verifier's work is four dot-products and one multiplication by the initial
matrix~$A$
(certifying the transposed to have a single matrix times $\lambda$-vector product). %
\end{proof}

\subsection{LDUP decomposition}\label{ssec:pldu}
With Protocol~\ref{cert:GRP}, when the matrix $A$ does not have
generic rank profile, any attempt to prove 
that it has generic rank profile will be detected w.h.p. (soundness).
However when it is the case, the verification will accept many possible vectors
$x,y,z$: any scaling of $z_i$ by $\alpha_i$ and $x_i,y_i$ by
$1/\alpha_i$ would be equally accepted for any non zero constants~$\alpha_i$.
This slack corresponds to our lack of specification of the diagonals
in the used LU decomposition.
Indeed, for any diagonal matrix with non zero elements, $LD\times
D^{-1}U$ is also a valid LU decomposition and yields $x,y$ and $z$
scaled as above. Specifying these diagonals is not necessary to prove
generic rank profileness, so we left it as is for this task.

However, for the determinant or the rank profile matrix certificates of
Sections~\ref{sec:det} and~\ref{sec:interractive:RPM}, we will need to
ensure that this scaling is independent from the choice of the vectors
$\phi,\psi,\lambda$. Hence we propose an updated protocol, where $L$ has to be
unit triangular, and the Prover 
has to first commit the 
main diagonal $D$ of $U$.

\newcommand{\xpart}{\overline{x}} %
\newcommand{\ypart}{\overline{y}}
\newcommand{\zpart}{\overline{z}}
For a non-singular upper triangular matrix $U$ with diagonal
$D=\text{Diag}(d_1,\ldots,d_n)$, the matrix $U_1=D^{-1}U$ is unit triangular.
Thus, for any
$
\psi=
\begin{bmatrix}\psi_1 \\ \widetilde{\psi}\end{bmatrix} %
\in \F^n\colon$ %
$U\psi=DU_1\psi=D\left(\psi+\begin{smatrix}\widetilde{U_1}\widetilde{\psi}\\0\end{smatrix}\right)$,
where $\widetilde{U_1} = (U_1 - I_n)_{\{1,\ldots,n-1\},\{2,\ldots,n\}}$ %
upper triangular in $\F^{(n-1)\times (n-1)}$. %
So the idea is that the Prover will commit $D$ beforehand,
and that within a generic rank profile certificate, the Verifier will only
communicate 
$\widetilde{\phi}, \widetilde{\psi}$ and $\widetilde{\lambda}$ to obtain
$\zpart=\widetilde{\lambda}^T\widetilde{L}$, $\xpart=\widetilde{U}_1\widetilde{\phi}$ and $\ypart=\widetilde{U}_1\widetilde{\psi}$,
where $\widetilde{L} = (L - I_n)_{\{2,\ldots,n\},\{1,\ldots,n-1\}}$ %
lower triangular in $\F^{(n-1)\times (n-1)}$. %
Then the Verifier will compute by himself the complete vectors.
This ensures that $L$ is unit triangular and that $U=DU_1$ with $U_1$ unit triangular.

Finally, if an invertible matrix does not have generic rank profile, we note
that it is also possible to incorporate the permutations, by
committing them in the beginning and reapplying them to the matrix
during the checks. 
The full certificate is given in Protocol~\ref{cert:pldu}.

\begin{protocol}[htbp]
  \centering
  {\small
  \begin{tabular}{|l c ll|}
    \hline\rule{0pt}{12pt}   
     Prover & & \hspace{15pt}Verifier&\\
      \multicolumn{4}{|c|}{\hspace{-25pt}$A \in \mathbb{F}^{n{\times}n}$ non-singular} \\
    \hdashline\rule{0pt}{12pt}\ignorespaces%
    $A=LDU_1 %
P$ & $\fxrightarrow{P,D}$ & \multicolumn{2}{l|}{$P\checks{\in}\mathcal{S}_n$, $D\checks{\in}\mathcal{D}_n(\F^*)$}
\\
$\widetilde{U_1} = (U_1 - I_n)_{\{1,\ldots,n-1\},\{2,\ldots,n\}}$ %
&&&\\
$\widetilde{L} = (L - I_n)_{\{2,\ldots,n\},\{1,\ldots,n-1\}}$ %
&&&\\
    \hdashline\rule{0pt}{12pt}   
    &&\multicolumn{2}{l|}{Choose $S\subset \F$}\\
    &%
&\multicolumn{2}{l|}{\textbf{for} $i$ from $n$ downto $2$:}\\
&$\vdots$&& %
\\
$\begin{bmatrix}\xpart&\ypart\end{bmatrix}
=\widetilde{U}_1
\begin{bmatrix}\widetilde{\phi}&\widetilde{\psi}\end{bmatrix}$
& $\fxleftarrow{\phi_i,\psi_i}$
&\hspace{2em} ($\phi_i,\psi_i$)%
&$\sample{S^2}$
\\
     & $\fxrightarrow{\xpart_{i-1},\ypart_{i-1}}$&  &\\ 
    $\zpart=\widetilde{\lambda}^T\widetilde{L}$ & $\fxleftarrow{\lambda_i}$ &\hspace{2em} $\lambda_i$&$\sample{S}$  \\
    & $\fxrightarrow{\zpart_{i-1}}$ &  &\\
   & $\vdots$& &\\
    \hdashline\rule{0pt}{12pt}    
    && $\phi_1,\psi_1,\lambda_1$&$\sample{S^3}$ \\ 
    && $\begin{bmatrix}x&y\end{bmatrix}$&$=\begin{bmatrix}\phi&\psi\end{bmatrix}+\begin{bmatrix}\xpart&\ypart\\0&0\end{bmatrix}$ \\
    && $z^T$&$=\left(\lambda^T+\begin{bmatrix}\zpart^T&0\end{bmatrix}\right)$ \\
    && $z^TD \begin{bmatrix}x&y\end{bmatrix} $&$\checks{=} (\lambda^T A)P^T \begin{bmatrix}\phi&\psi\end{bmatrix}$ \\
    \hline
  \end{tabular}
  \caption{\mbox{LDUP} decomposition (linear communication)}
  \label{cert:pldu}
  }
\end{protocol}

\begin{theorem}\label{thm:pldu}

The Protocol~\ref{cert:pldu} requires less than $8n$ extra communications.
The computational cost for the Prover is $\GO{n^\omega}$ and the Verifier cost
is bounded by $\mu(A)+12n+o(n)$. 
The protocol is perfectly complete and fails the verification for a non generic
rank profile matrix
$AP^{-1} = AP^T$ %
with probability $\ge (1-\frac{1}{|S|})^{2n}$.
\end{theorem}
\begin{proof}
If the Prover is honest, then $A=LUP=LDU_1P$, so that for any choice of $\lambda$
and $\psi$ we have:
$\lambda^TAP^{T}\psi=\lambda^T LDU_1 \psi$, that is:
\[\begin{split}
z^T D y& = (\lambda^T+\begin{bmatrix}\zpart^T&0\end{bmatrix})D\left(\psi+\begin{bmatrix}\ypart\\0\end{bmatrix}\right)\\
& =
\begin{bmatrix}\lambda_1&\widetilde{\lambda}^T\end{bmatrix}
\left(I+\begin{smatrix}0&0\\\widetilde{L}&0\end{smatrix}\right)D\left(\begin{smatrix}0&\widetilde{U_1%
}\\0&0\end{smatrix}+I\right)\begin{bmatrix}\psi_1\\\widetilde{\psi}\end{bmatrix}.
\end{split}\]

The same is true for $\lambda$ and $\phi$, so that the protocol is
perfectly complete.

Now, the last part of the Protocol~\ref{cert:pldu} is
actually a verification that $AP^T$ has generic rank
profile, in other words that there exists lower and upper triangular
matrices $L^*$ and $U^*$ such that $AP^T=L^*U^*$. This
verification is sound by Theorem~\ref{th:GRP}. 
Next, the multiplication by the diagonal $D$ is performed by the
Verifier, in order to be %
actually convinced that there exists lower and
upper triangular matrices $L^*$ and $U_1^*$ such that
$AP^T=L^* D U_1^*$.   
Finally, the construction of the vectors with the form
$a+\begin{smatrix}\widetilde{b}\\0\end{smatrix}$ is also done by the
Verifier, in order to have in fact %
a guarantee that $L^*$ and $U_1^*$ are
unit triangular. 

Overall, if the matrix
$AP^T$ %
does not have generic rank profile, the Verifier will
catch him with the probability of Theorem~\ref{th:GRP}. 

Finally, for the complexity bounds, the extra communications are:
one permutation matrix $P$, a diagonal matrix $D$ and $6$ vectors 
$\widetilde{\lambda}$, $\widetilde{\phi}$, $\widetilde{\psi}$ and 
$\zpart$, $\xpart$ and $\ypart$. 
That is $n$ non-negative integers lower
than $n$ and $6(n-1)+n$ field elements. The arithmetic computations of the
Verifier are one multiplication by a diagonal matrix, $3$ vector sums, $4$
dot-products and one vector-matrix %
multiplication by $A$ (for
$(\lambda^TA)$), that is $n+3(n-1)+4(2n-1)$. 
\end{proof}

We do not need the following %
fact
to show that Protocol~\ref{cert:pldu} correctly
verifies generic rank profileness, but furthermore, this protocol actually gives
some guarantees on the actual values of $D$ and $x,y,z$:
\begin{proposition}\label{cor:uniqueXY}
  Let $S$ be a finite subset of $\F$ %
in Protocol~\ref{cert:pldu}, 
  if $AP^T$ is not in generic rank profile, 
  or else 
  if the committed $D$ does not correspond to the unique decomposition $AP^T=LDU_1$ 
  or
  $\begin{bmatrix}x&y\end{bmatrix}\neq{}U_1\begin{bmatrix}\phi&\psi\end{bmatrix}$ 
  or 
  $z^T\neq{}\lambda^TL$,
  then the verification will fail with probability $\ge (1-\frac{1}{|S|})^{2n}$,
  and therefore Protocol~\ref{cert:pldu} is sound.
\end{proposition}
\begin{proof} For a dishonest Prover,
  either
  \begin{compactenum}[(i)]
  \item \label{case:GRP} $AP^T %
$ is not in generic rank profile, then Protocol~\ref{cert:pldu} will detect it with the probability of Theorem~\ref{thm:pldu};
  \item \label{case:discr} or she could still try,
to send modified vectors $\xpart$, $\ypart$, $\zpart$ or diagonal $D$.
  \end{compactenum}
  Let then 
  $D^*$, 
  $x^*=\phi+\begin{smatrix}\xpart^*\\0\end{smatrix}=U_1\phi$, 
  $y^*=\psi+\begin{smatrix}\ypart^*\\0\end{smatrix}=U_1\psi$ and
  $z^*=\begin{smatrix}\zpart^*\\0\end{smatrix}+\lambda=L^T\lambda$ be the
  correct expected diagonal and vectors.
  Let also $i_0\le n$ %
be the largest index such that there is at least one discrepancy in
  $d_{i_0}$, $\xpart_{i_0}$, $\ypart_{i_0}$ or $\zpart_{i_0}$ that makes at least one of 
  them respectively different from $d^*_{i_0}$, $\xpart^*_{i_0}$, $\ypart^*_{i_0}$ or %
  $\zpart^*_{i_0}$
  ($\xpart_n=\xpart^*_n = 0,\ypart_n=\ypart^*_n = 0,\zpart_n=\zpart^*_n = 0$ by default). %
  Then $H_i$ of (\ref{eq:predicate}) %
is true for all $i$ such that $n\ge i>i_0$, and thus in particular
  $H_{i_0+1}$ is true
($H_{n+1}$ is true by default). %
  Now, $H_{i_0}$ is also true if and only if we have both:
  \begin{equation}\label{eq:zixi}
    \begin{cases}
      z_{i_0}d_{i_0}x_{i_0}=z^*_{i_0}d^*_{i_0}x^*_{i_0},%
\\
      z_{i_0}d_{i_0}y_{i_0}=z^*_{i_0}d^*_{i_0}y^*_{i_0}.
    \end{cases}
  \end{equation}
  Indeed, $H_{i_0}$ is $(z\supbr{i_0\dots n})^T %
  D\supbr{i_0\dots n} x\supbr{i_0\dots n}=({z^*}\supbr{i_0\dots n})^T %
  D\supbr{i_0\dots n}{x^*}\supbr{i_0\dots n}$ %
  and similarly $H_{i_0+1}$ is $(z\supbr{i_0+1\dots n})^T  %
  D\supbr{i_0+1\dots n} x\supbr{i_0+1\dots n}=({z^*}\supbr{i_0+1\dots n})^T %
  D\supbr{i_0+1\dots n} {x^*}\supbr{i_0+1\dots n}$.
  Further, Equations~(\ref{eq:zixi}),
  with
  $a=\zpart^*_{i_0}d^*_{i_0}\xpart^*_{i_0}-\zpart_{i_0}d_{i_0}\xpart_{i_0}$,
  and 
  $b=\zpart^*_{i_0}d^*_{i_0}\ypart^*_{i_0}-\zpart_{i_0}d_{i_0}\ypart_{i_0}$,
  is equivalent to:
  \begin{equation}\label{eq:lDphilpsi}
    \begin{cases}
      \lambda_{i_0}\phi_{i_0}(d_{i_0}-d^*_{i_0})+
      \lambda_{i_0}(d_{i_0}\xpart_{i_0}-d^*_{i_0}\xpart^*_{i_0})+
      \phi_{i_0}(d_{i_0}\zpart_{i_0}-d^*_{i_0}\zpart^*_{i_0})-a=0,%
\\
      \lambda_{i_0}\psi_{i_0}(d_{i_0}-d^*_{i_0})+
      \lambda_{i_0}(d_{i_0}\ypart_{i_0}-d^*_{i_0}\ypart^*_{i_0})+
      \psi_{i_0}(d_{i_0}\zpart_{i_0}-d^*_{i_0}\zpart^*_{i_0})-b=0.%
    \end{cases}
  \end{equation}
  However, $\lambda_{i_0},\phi_{i_0},\psi_{i_0}$ are chosen by the Verifier after
  $d_{i_0}$, $\xpart_{i_0}$, $\ypart_{i_0}$ and $\zpart_{i_0}$ have been committed.
  Hence, on the one hand, if $d_{i_0}\neq d^*_{i_0}$
then the coefficient of $\lambda_{i_0}$
in one of the two polynomials is not equal to $0$ %
for a random $\phi_{i_0}$ with probability $\ge 1-1/|S|^2$ %
and then that %
polynomial
does not vanish for a random $\lambda_{i_0}$ with probability
$\ge (1-1/|S|^2)(1-1/|S|)$, %
  based on the random choices made by the Verifier, and 
$H_{i_0}$ is violated.

  On the other hand, if $d_{i_0}=d^*_{i_0}\neq 0$, they can be removed from
  Equations~(\ref{eq:lDphilpsi}) which then simplifies
(for $i_0 < n$) %
as: 
  \begin{equation}\label{eq:lphilpsi}
    \begin{cases}
      \lambda_{i_0}(\xpart_{i_0}-\xpart^*_{i_0})+\phi_{i_0}(\zpart_{i_0}-\zpart^*_{i_0})-(\zpart^*_{i_0}\xpart^*_{i_0}-\zpart_{i_0}\xpart_{i_0})=0,%
\\ %
      \lambda_{i_0}(\ypart_{i_0}-\ypart^*_{i_0})+\psi_{i_0}(\zpart_{i_0}-\zpart^*_{i_0})-(\zpart^*_{i_0}\ypart^*_{i_0}-\zpart_{i_0}\ypart_{i_0})=0.%
    \end{cases}%
  \end{equation}
  When there is at least one discrepancy with the expected vector coefficients, then
  Equations~(\ref{eq:lphilpsi}) can be considered as 
  $2$ polynomials that are not simultaneously identically zero. %
  Thus they both vanish with probability $\leq 1/|S|$ %
  based on the random choices made by the Verifier.
  $H_{i_0}$ is thus false with probability $\ge (1-1/|S|)$. 
  As in the proof of Theorem~\ref{th:GRP}, this propagates with high probability, %
to $H_1$ and
  the dishonest Prover is detected with probability 
  $\ge (1-1/|S|)^{2(n-1)} (1-1/|S|^2)(1-1/|S|) %
\ge (1-1/|S|)^{2n}$.
 
  Overall, both (\ref{case:GRP}), $AP^T %
$ is not GRP, or
  (\ref{case:discr}), $AP^T%
$ is GRP but some diagonal or vector elements is wrong, are detected with
  probability $\ge (1-1/|S|)^{2n}$.
\end{proof}

\section{Linear communication interactive certificates}
\label{sec:interractive}

In this section, we give linear space communication certificates for
the determinant,
the column/row rank profile of a matrix, and for the rank profile matrix.

\subsection{Linear communication certificate for the
  determinant}\label{sec:det}
Existing certificates for the determinant are either optimal for the Prover in
the dense case, using the strategy of \cite[Theorem~5]{kns11} over a PLUQ
decomposition, but quadratic in communication; or linear in communication, using
\cite[Theorem~14]{jgd:2016:gammadet}, but using a reduction to the
characteristic polynomial. 
In the sparse case the determinant and the characteristic polynomial both
reduce to the same minimal polynomial computations and therefore the latter
certificate is currently optimal for the Prover. Now in the dense case, while
the determinant and characteristic polynomial both reduce to matrix
multiplication, the determinant, via a single PLUQ decomposition is more
efficient in practice~\cite{Pernet:2007:charp}. 
Therefore, we propose here an alternative in
the dense case: use only one PLUQ decomposition for the Prover while keeping
linear extra communications and $O(n)+\mu(A)$ operations for the Verifier.
The idea is to extract the 
information of a LDUP decomposition without communicating it: 
one uses Protocol~\ref{cert:pldu} for $A=LDUP$ with $L$ and $U$
unitary, but kept on the Prover side, and then the Verifier only has to compute
$Det(A)=Det(D)Det(P)$, with $n-1$ additional field operations.%

\begin{corollary}\label{thm:det}
For an $n{\times}n$ matrix, there exists a sound and perfectly complete protocol
for the determinant over a field using less than $8n$ extra
communications and with computational cost for the Verifier bounded by
$\mu(A)+13n+o(n)$. 
\end{corollary}
As a comparison, the protocol of \cite[Theorem~14]{jgd:2016:gammadet} reduces
to {\sc{CharPoly}} instead of PLUQ for the Prover, requires $5n$ extra
communications and $\mu(A)+13n+o(n)$ operations for the Verifier as well. 
Also the new protocol requires $3n$ random field elements for a field larger
than $2n$, where that of \cite[Theorem~14]{jgd:2016:gammadet} requires $3$
random elements but a field larger than $n^2$. Finally the new protocol requires
$O(n)$ rounds when $2$ are sufficient in \cite[Theorem~14]{jgd:2016:gammadet}.

For instance, using the routines shown in Table~\ref{tab:io} (one
matrix-vector multiplication with a dense matrix is denoted \texttt{fgemv}),
the determinant of an 
$50k{\times}50k$ random dense
matrix can be computed in about
24~minutes,
where with the certificate of Protocol~\ref{cert:pldu}, the overhead of the Prover is less than 
5s %
and the Verifier time is about
1s.%

Computations use the
FFLAS-FFPACK library~\cite{fflas:2017}
on a single Intel Skylake core @3.4GHz, while we measured some communications
between two workstations over an Ethernet Cat. 6, @1Gb/s network cable.
We see that a linear communication cost can be masked by a quadratic number of
computations, when a quadratic communication cost could be up to two orders of
magnitude worse.

\begin{table}[htbp]
  \centering
  \begin{tabular}{lrrr}
    \toprule
    Dimension		& $2k$& $10k$&$50k$ \\
    \midrule
    PLUQ		& 0.28s	& 17.99s	& 1448.16s\\
    {\sc{CharPoly}}	& 1.96s & 100.37s & 8047.56s\\
    \midrule
    Linear comm.	& 0.50s	& 0.50s	& 0.50s\\
    Quadratic comm.	& 1.50s	& 7.50s	& 222.68s\\
    \midrule
    \texttt{fgemv}	& 0.0013s& 0.038s& 1.03s\\
    \bottomrule
  \end{tabular}
  \caption{Communication of 64 bit words versus computation modulo
    $131071$}\label{tab:io}
\end{table}

\subsection{Column or row rank profile certificate}
\label{sec:interractive:CRP}
In Protocols~\ref{cert:upper_rank} and~\ref{cert:lower_rank}, we first
recall the two linear time and space certificates for an upper and a
lower bound to the rank that constitute a rank certificate. 
We present
here the variant sketched in~\cite[\S~2]{eberly15} of the certificates
of~\cite{dk14}.  
An upper bound $r$ on the rank is certified by the capacity for the
Prover to generate any vector sampled from the image of $A$ by a linear
combination of $r$ column of $A$ ($\|\gamma\|_0 %
                                  $ denotes the
Hamming weight of the vector $\gamma$). 
A lower bound $r$  is certified by the capacity for the Prover to recover the
unique coefficients of a linear combination of $r$ linearly independent columns
of $A$. $\LINSYS(r)$ denotes a complexity bound for solving a linear system of
rank $r$ by the Prover.

\begin{protocol}[htbp]
    \centering
    \begin{tabular}{|l c l|}
        \hline
        Prover &  & Verifier \\
        & $ A \in \mathbb{F}^{m \times n} $ &  \\
        \hdashline\rule{0pt}{12pt}
        $R %
        $ s.t. $ \rank(A) \leq R $ & $\fxrightarrow{\text{$R$}} $ & \\
        \hdashline\rule{0pt}{12pt}
        & & Choose $S \subset \mathbb{F}$ \\
        & $\fxleftarrow{\text{$w$}}$ & $ v \sample{S^n}, w = Av $ \\
        \hdashline\rule{0pt}{12pt}
        $A\gamma = w$ & $\fxrightarrow{\text{$\gamma$}}$ & $\|\gamma\|_0 %
                        \stackrel{?}{=}R$\\
        & & $A\gamma \stackrel{?}{=} w $\\
        \hline
    \end{tabular}
    \caption{Upper bound on the rank of a matrix}
    \label{cert:upper_rank}
\vspace{\spaceafterprotocol}\end{protocol}

\begin{theorem}
    \label{th:upper_rank}
    Let $A \in \mathbb{F}^{m \times n}$, and let $S$ be a finite subset
    of $\mathbb{F}$.
    The interactive certificate~\ref{cert:upper_rank} of
    an upper bound for the rank of $A$ is sound, with probability larger than
    $1-\frac{1}{|S|}$, perfectly complete, occupies $m+n$ communication space,
    can be computed in $\LINSYS(r)$ and verified in $2\mu(A) + n$ time.
\end{theorem}

\begin{protocol}[htbp]
    \centering
    \begin{tabular}{|l c l|}
        \hline
        \multicolumn{1}{|l}{Prover} &  & Verifier \\
        & $A\in\mathbb{F}^{m{\times}n}$ &  \\
        \hdashline\rule{0pt}{12pt}
        $\mathcal{J}=(c_1,.., c_\rho %
                     )$ indep. cols of A &
        $\fxrightarrow{\text{$c_1, .., c_\rho$}} $ & \\
        \hdashline\rule{0pt}{12pt}
        & & Choose $S \subset \mathbb{F^*}$ \\
        & $\fxleftarrow{v}$ &  $\alpha = \begin{dcases*}
          \alpha_{c_j} \sample{S}  \\ 0\text{ otherwise} \end{dcases*} $ \\
        & & $v = A\alpha$\\
        \hdashline\rule{0pt}{12pt}
        Solve $A\beta = v$ & $\fxrightarrow{\text{$\beta$}}$ & $\beta\checks{=}\alpha$\\
        \hline
    \end{tabular}
    \caption{Lower bound on the rank of a matrix}
    \label{cert:lower_rank}
\vspace{\spaceafterprotocol}\end{protocol}
\begin{theorem}
    \label{th:lower_rank}
    Let $A \in \mathbb{F}^{m \times n}$, and let $S$ be a finite subset
    of $\mathbb{F}$.
    The interactive certificate~\ref{cert:lower_rank} of a lower bound
    for the rank of $A$ is sound, , with probability larger than
    $1-\frac{1}{|S|}$, perfectly complete and occupies $m + 2r$ communication
    space, can be computed in $\LINSYS(r)$ and verified in $\mu(A)+r$ operations.
\end{theorem}
Note that the communication in Protocol~\ref{cert:lower_rank} involve sending
$r$ indices for $\mathcal{J}$, then $m$ field elements for vector $v$, and only
$r$ field elements for vector $\beta$, as it has only $r$ non-zero coefficients
which positions are already indicated by $\mathcal{J}$. Hence the total
communication cost is $m+2r$.

We now consider a column rank profile certificate: the Prover is given a matrix $A$, and answers
the column rank profile of $A$, $\mathcal{J} = (c_1, \dots, c_r)$. %
In order to certify this column rank profile, we need to certify two properties: 
\begin{compactenum}
\item the columns given by $\mathcal{J}$ are linearly independent; \label{CRP:step:indep}
\item the columns given by $\mathcal{J}$ form the lexicographically smallest set \label{CRP:step:minimal}
  of independent columns of $A$.
\end{compactenum}

Property~\ref{CRP:step:indep} is verified by Certificate~\ref{cert:lower_rank},
as 
it checks whether a set of columns are indeed linearly independent.
Property~\ref{CRP:step:minimal} could be certified by successive applications of Certificate~\ref{cert:upper_rank}: at step $i$, checking that the rank of $A_{*, (0, \dots, c_i - 1)}$
is at most $i-1$ would certify that there is no column located between $c_{i-1}$ and $c_i$ in $A$ 
which increases the rank of $A$. Hence, it would prove  the minimality of $\mathcal{J}$.
However, this method requires $O(nr)$  communication space.

Instead, one can reduce the communication by seeding all challenges from
a single $n$ dimensional vector, and by compressing the responses with a random
projection.
The right triangular equivalence certificate  plays here a central role,
ensuring the lexicographic minimality of $\mathcal{S}$.
More precisely, the Verifier chooses a vector $v \in \F^{n}$ uniformly at random
and sends it to the Prover.
Then, for each index $c_k\in\mathcal{S}$ the Prover computes the linear
combination of the first $c_k-1$  columns of $A$ using the first $c_k-1$
coefficients of $v$ and has to prove that it can be generated from the $k-1$
columns $c_1,\dots,c_{k-1}$. This means, find  a vector $\gamma^{(k)}$ solution to the system:
$$
\begin{bmatrix}
  A_{*,c_1} &   A_{*,c_2} & \dots &   A_{*,c_{k-1}}  
\end{bmatrix}
\gamma^{(k)}
=
A 
\begin{smatrix}
  v_1\\\vdots\\ v_{c_k-1}\\0\\\vdots
\end{smatrix}.
$$

Equivalently, find an upper triangular matrix $\Gamma$ such that:
\begin{equation}\label{eq:Vdef}%
\begin{bmatrix}
  A_{*,c_1} &   A_{*,c_2} & \dots &   A_{*,c_{r-1}}  
\end{bmatrix}
\Gamma =
A
\underbrace{\begin{smatrix}
  v_1&v_1 & \cdots & \cdots&v_1\\
  \vdots&\vdots&  \vdots& \vdots&\vdots\\
  v_{c_1-1}&\vdots & \vdots&\vdots&\vdots\\
  0 & v_{c_2-1}&\vdots& \vdots&\vdots\\
  0 & 0 & \ddots & \vdots&\vdots\\
  0 & 0 & 0 & v_{c_r-1}&\vdots\\
  0 & 0 & 0 & 0&v_{n}
\end{smatrix}}_{V}.
\end{equation}

Note that $V=\text{Diag}(v_1,\dots,v_n)W$ where $W = [\mathds{1}_{i<c_{j+1}}]_{i,j}$ (with $c_{r+1}=n+1$ by convention)
In order to avoid having to transmit the whole $r\times r$ upper triangular matrix
$\Gamma$, the Verifier only checks a random projection $x$ of it, using the
triangular equivalence Certificate~\ref{cert:triangular_equivalence}.
We then propose the certificate in Protocol~\ref{cert:crp}.
\begin{protocol}[htbp]
  \centering
  {\def\arraystretch{1.2}
    \begin{tabular}{|l c l|}
        \hline
        Prover & & Verifier \\
        & $A \in \mathbb{F}^{m\times n}$ & \\
        \hdashline\rule{0pt}{12pt}
        $\mathcal{J}=(c_1,.., c_r)$ CRP of $A$ &   $\begin{array}{|c|}
        \hline
        \text{\scriptsize Protocol~\ref{cert:lower_rank} on $A$} \\
             \hdashline%
             \xrightarrow{\mathcal{J} = CRP(A) }\\
             \rank{A} \checks{\geq} r \\
             \hline
 \end{array}$

&\\
        \hdashline%
        & & Choose $S \subset \mathbb{F}$ \\
        & $\xleftarrow{~~ v ~~ }$ & $v \sample{S^{n}}$\\
        $V=\text{Diag}(v_1,\ldots,v_n %
                      )W$
        (see (\ref{eq:Vdef})) %
        &&    $W = [\mathds{1}_{i<c_{j+1}} ]\in\F^{n{\times}r}
        $\\

        $\Gamma$ upper tri.  s.t. $A_{*,\{c_1,..,c_r\}}\Gamma = AV$ & &\\
        \hdashline%
        $y = \Gamma x$ 
        & $ \xleftrightarrow{x \text{  (Cert.~\ref{cert:triangular_equivalence})  } y}$& $x \sample{S^{r}}$ \\
        \hdashline%
 && $z=\text{Diag}(v_1,\ldots,v_n %
                             )(Wx)$\\
 && $z_{c_j}=z_{c_j}{-}y_j$ for $j=1..r$\\
 && $Az\checks{=}0$\\
        \hline
    \end{tabular}
    \caption{Certificate for the column rank profile}
    \label{cert:crp}
    }
\vspace{\spaceafterprotocol}\end{protocol}

\begin{theorem}
    For $A \in \F^{m \times n}$ and $S \subset \mathbb{F}$,
    certificate~\ref{cert:crp} is sound, with probability larger than
    $1-\frac{1}{|S|}$, 
    perfectly complete, with a Prover
    computational cost bounded by $O(mnr^{\omega - 2})$,
    a communication space complexity bounded by $m+n+4r$
    and a Verifier cost bounded by $2\mu(A)+n+3r$.
\end{theorem}

\begin{proof}
    If the Prover is honest, the protocol corresponds first to an application of
    Theorem~\ref{th:lower_rank} to certify that $\mathcal{J}$ is a set of
    independent columns. This certificate is perfectly complete.
    Second the protocol also uses challenges from
    Certificate~\ref{cert:upper_rank}, 
    which is perfectly complete, together with
    Certificate~\ref{cert:triangular_equivalence}, which is perfectly complete
    as well. The latter certificate is used on $A_{*, \mathcal{J}}$, a
    regular submatrix, as $\mathcal{J}$ is a set of independent columns of $A$.
    The final check then corresponds to
    $A(D(Wx)) -A_{*,\{c_1,..c_r\}}y \stackrel{?}{=} 0$ and, overall,
    Certificate~\ref{cert:crp} is perfectly complete. 

    If the Prover is dishonest, then either the set of  columns in $\mathcal{J}$ are  not
    linearly 
    independent, which will be caught by the Verifier with probability at least 
    $1 - \frac{1}{|S|}$, from Theorem~\ref{th:lower_rank},
    or
    $\mathcal{J}$ is not lexicographically minimal, or the rank of $A$ is not $r$.
    If the rank is wrong, it will not be possible for the Prover 
to find a
    suitable $\Gamma$. This will be caught by the Verifier 
verifier with probability 
    $1-\frac{1}{|S|}$, from %
    Theorem~\ref{th:triangular_equivalence}.
    Finally, if $\mathcal{J}$ %
    is not lexicographically minimal, there exists at least one column 
    $c_k \notin    \mathcal{J}, c_i < c_k < c_{i+1}$ for some fixed 
    $i$ such that $\{c_1, \dots, c_i\} \cup \{c_k\}$ form a set of linearly
    independant columns of~$A$.
    This means that $\rank(A_{*, 1, \dots, c_{i+1} - 1}) = i+1$, whereas it was
    expected to be $i$. 
    Thus, the Prover 
cannot reconstruct a suitable triangular $\Gamma$ and this
    will be detected by the Verifier 
also with probability $1-\frac{1}{|S|}$, as
    shown in Theorem~\ref{th:triangular_equivalence}.

    The Prover's time complexity is that of computing a $PLUQ$ decomposition of
    $A$.    
    The transmission of $v, x$ and $y$ yields a 
    communication cost of $n+2r$, which adds up to the $m+2r$ communication cost
    of Protocol~\ref{cert:lower_rank}. Finally, in addition to
    Protocol~\ref{cert:lower_rank}, the Verifier computes $Wx$ as a prefix sum with $r-1$
    additions, multiplies it by $D$, then subtracts 
$y_i$ at the $r$
    correct positions and finally multiplies by $A$ for a total cost bounded by
    $2\mu(A)+n+3r-1$.
\end{proof}

\subsection{Rank profile matrix certificate}
\label{sec:interractive:RPM}

We propose an interactive certificate for the rank profile matrix based
on~\cite[Algorithm~4]{jgd:2017:bruhat}: first computing the row and column support of the
rank profile matrix, using Certificate~\ref{cert:crp} twice for the row and
column rank profiles, then computing the rank profile matrix of the invertible
submatrix of $A$ lying on this grid.

In the following we then only focus on a certificate for the rank profile
matrix of an invertible matrix. It relies on an LUP decomposition that reveals
the rank profile matrix. From Theorem~\ref{lem:echelonized}, this is the case if
and only if $P^TUP$ is upper triangular. 
Protocol~\ref{cert:RPM:linearcomm} thus gives an interactive certificate that combines
Certificate~\ref{cert:pldu} for a LDUP decomposition with a certificate that $P^TUP$ is upper
triangular. The latter is achieved by
Certificate~\ref{cert:triangular_equivalence} showing that $P^T$ and $P^TU$ are
left upper triangular equivalent, but since $U$ is unknown to the Verifier, the
verification is done on a random right projection with the vector $\phi$  used in Certificate~\ref{cert:pldu}.
  \begin{protocol}[htbp]\center
    {%
      \centering
  {\def\arraystretch{1.2}
    \begin{tabular}{|l c l|}
        \hline
        Prover & & Verifier \\
        \multicolumn{3}{|c|}{$A \in \mathbb{F}^{n\times n}$ invertible} \\
        \hdashline%
          $A=LDUP$, with $P=\RPM{A}$
        &
        $\xrightarrow{\ P,D \ }$ &$P \checks{\in}\mathcal{S}_n$, $D\checks{\in}\mathcal{D}_n(\F)$\\
        \hline
        \multicolumn{3}{|c|}{Protocol~\ref{cert:triangular_equivalence} :
          {\small $P^T$ and $P^TU$ are left up. tri. equiv. with random proj.}}\\
        \hdashline%
        $\overline{U}= P^TUP$ & $\xrightarrow{\overline{U} \text{ is upper tri.}}$&\\
        \hdashline%
        &&Choose $S\subset \F$\\ %
        &$\xleftarrow{e_1,\ldots,e_n %
                     }$&\textbf{for } $i=1,\ldots, %
                                             n$, $e_i\sample{S}$\\
        $f^T= e^T\overline{U}$ &  $\xrightarrow{f_1,\ldots,f_n %
                                                        }$& \\
        &  %
        $\begin{array}{|c|}
          \hline
          \text{\scriptsize Protocol~\ref{cert:pldu} on $A$ %
               }\\
               \hdashline%
               \xleftarrow{ \begin{smatrix} \widetilde{\phi}&\widetilde{\psi} \end{smatrix}
               } \\
               \xrightarrow{\begin{smatrix}\widetilde{x}&\widetilde{y}\end{smatrix}} \\
               \hline
        \end{array}$
        & $\phi,\psi\sample{S^n}$   \\
        && Now $
        \begin{bmatrix} x&y \end{bmatrix}$ is $U \begin{bmatrix} \phi&\psi \end{bmatrix}
        $\\
        \hdashline
        && $ e^TP^T x
        \checks{=}  f^T P^T   \phi
        $  \\
        \hline
    \end{tabular}
    \caption{Rank profile matrix of an invertible matrix}
    \label{cert:RPM:linearcomm}
  }
  }
\vspace{\spaceafterprotocol}\end{protocol}
  \begin{theorem}
    Protocol~\ref{cert:RPM:linearcomm} is sound, with probability 
    $\ge (1-\frac{1}{|S|})^{2n}$, and perfectly complete. The Prover cost is $\GO{n^\omega}$
    field operations, the communication space is bounded by $10n$ and the
    Verifier cost is bounded by $\mu(A)+16n+o(n)$.
  \end{theorem}

  \begin{proof}
  If the Prover is dishonest and  $\overline{U}=P^TUP$ is not upper triangular,
then let $(i,j)$ be the lexicographically minimal coordinates such that $i>j$ and
$\overline{U}_{i,j}\neq 0$.
Now either
$\begin{bmatrix} x&y  \end{bmatrix} \neq U \begin{bmatrix} \phi&\psi  \end{bmatrix}$,
and the verification will then fail to detect it with probability less than $(1-\frac{1}{|S|})^{2n}$, from
Proposition~\ref{cor:uniqueXY}.
Or one can write
$e^TP^T  x- f^TP^T \phi = (e^T\overline{U}-f^T) P\phi=0$.
  If
  \begin{equation}\label{eq:qtuq}
    e^TP^TUP-f^T=0.
  \end{equation}
  is not satisfied, then a random $\phi$ will fail to detect it
  with probability less than $\frac{1}{|S|}$, since
$e,\overline{U}$ and $f$ are set before  choosing for $\phi$.
 At the time of committing 
$f_j$, the value of $e_i$
  is still unknown, hence $f_j$ is constant in the symbolic variable
  $E_i$. Thus the $j$-th coordinate in~\eqref{eq:qtuq} is a nonzero polynomial
  in $E_j$ and therefore vanishes with probability $1/|S|$ when sampling the
  values of $e$ uniformly. Hence, overall if $P^TUP$ is not upper triangular, the
  verification will detect it with probability $\ge (1-\frac{1}{|S|})^{2n}$.

  The Verifier's cost is that of Protocol~\ref{cert:pldu} with two additional
  dot products for the last step, which is $\mu(A)+16n+o(n)$. Similarly, the communication
  cost is that of Protocol~\ref{cert:pldu} plus the size of $e$ and $f$ for a
  total of $10n$. The Prover remains unchanged.
 \end{proof}

Finally, we use~\cite[Algorithm~4]{jgd:2017:bruhat} to certify the rank profile matrix
of any matrix, even a singular one. To do so, we need to verify the row rank profile and the 
column rank profile of the input matrix, which can be done with two applications of 
Certificate~\ref{cert:crp}. Then, we certify the rank profile matrix of the $r \times r$
selection %
of lexicographically minimal independent rows and columns we obtained before.
This is done by an application of Certificate~\ref{cert:RPM:linearcomm}.
We now define $\mathcal{E}_{m, \{i_1, \dots, i_n\}}$ as the $m \times n$ matrix
whose $j$-th column is the $i_j$-th vector of the $m$-dimensional canonical
basis. This certificate is detailed in Protocol~\ref{cert:RPM:full}, in the case where $m \leq n$.
If $n < m$, one should first apply Protocol~\ref{cert:crp} on $A$ to compute its column rank profile, and 
then apply the verification steps of the same protocol for the row rank profile of $A$. The application of
Protocol~\ref{cert:RPM:linearcomm} remains unchanged.

\begin{theorem}
  Protocol~\ref{cert:RPM:full} is sound, with probability 
  $\ge (1-\frac{1}{|S|})^{2n}$, and perfectly 
  complete. The Prover cost is $\GO{mnr^{\omega-2}}$ field operations, the
  communication space is bounded by 
  $m+n+\min(m,n)+17r$ and the Verifier cost is bounded by $4\mu(A)+m+n+21r$.
\end{theorem}

\begin{proof}
    If the Prover is honest, $\mathcal{I}$ is the row rank profile of A and $\mathcal{J}$ is the 
    column rank profile of $A$. Then, the application of Protocol~\ref{cert:RPM:linearcomm} will 
    output the correct rank profile matrix of $A_{\mathcal{I}, \mathcal{J}}$ which will lead the Verifier
    to the correct rank profile matrix of $A$, as described in~\cite[Theorem~37]{jgd:2017:bruhat}.
    Note that one only needs to verify the lower bound on the rank of $A$ once, which is why 
    Certificate~\ref{cert:crp} is fully executed once, while the second run only verifies that the committed 
    rank profile is a rank profile indeed.

    Now, for the soundness, Prover has a probability $\ge 1-1/|S|$ to be caught when
    cheating while running
    Certificate~\ref{cert:crp}, and a probability $\ge (1-\frac{1}{|S|})^{2n}$ to be caught when
    cheating while running 
    Certificate~\ref{cert:RPM:linearcomm}. Overall, this makes a probability
    $\ge (1-\frac{1}{|S|})^{2n}$ for the
    Verifier to catch a cheating Prover during the execution of Certificate~\ref{cert:RPM:full}.

    For the complexity, Prover time complexity is bounded by the complexity of performing a PLUQ 
    decomposition of the input matrix, $\GO{mnr^{\omega - 2}}$. The Verifier complexity is the one
    of one full application of Protocol~\ref{cert:crp} and one application of Protocol~\ref{cert:crp}
    without applying 
Protocol~\ref{cert:lower_rank}, which makes $3\mu(A) + n + m + 5r$, plus one
    application of Protocol~\ref{cert:RPM:linearcomm} over an $r \times r$ matrix for a cost of 
    $\mu(A) + 16r + o(r)$, the computation of $\mathcal{R}_A$ only consists of memory operations, 
    hence a total cost of $4\mu(A) + m + n + 21r + o(r)$ field operations.
    Communication space is computed as follows: a full application of Protocol~\ref{cert:crp} on
    $A$ if $m \geq n$, on $A^T$ otherwise, an application of the same Protocol without the underlying
    Protocol~\ref{cert:lower_rank} which makes $n + m + min(m, n) + 7r$ and the same application of 
    Protocol~\ref{cert:RPM:linearcomm} as above, for a cost of $10r$, hence a total communication space
    of $m + n + min(m, n) + 17r$.
\end{proof}

\begin{protocol}[htbp]
  {%
    \centering
{\def\arraystretch{1.2}
    \begin{tabular}{|l c l|}
      \hline
      Prover & & Verifier \\
      \multicolumn{3}{|c|}{$A \in \mathbb{F}^{m\times n}$ (assuming $m\leq n$ w.l.o.g.)}\\
      \hdashline%

      &$\begin{array}{|c|}
        \hline
           \text{\scriptsize Protocol~\ref{cert:crp} on $A^T$} \\
             \hdashline%
             \xrightarrow{\mathcal{I} = RRP(A) \in [\![1, m ]\!]^r} \\
             \hline
      \end{array}$&\\
      && Now, $\mathcal{I} = RRP(A)$\\
      && and $\rank(A) \geq r$ \\\hdashline
      &&\\
      $\mathcal{J} = CRP(A) = (c_1, \dots, c_r)$ & $\xrightarrow{\mathcal{~~~~J~~~~}}$ & \\
      & & Choose $S \subset \mathbb{F}$ \\
      & $\xleftarrow{~~~~v~~~~}$ & $v \sample{S^{n}}$\\
      $V=\text{Diag}(v_1,\ldots,v_n %
                    )W$ (see (\ref{eq:Vdef})) %
      &&    $W = [\mathds{1}_{i<c_{j+1}} ]\in\F^{n{\times}r}
      $\\

      $\Gamma$ upper tri.  s.t. $A_{*,\{c_1,..,c_r\}}\Gamma = AV$ & &\\
      \hdashline%
      $y = \Gamma x$ 
      & $ \xleftrightarrow{x \text{  (Cert.~\ref{cert:triangular_equivalence})  } y}$& $x \sample{S^{r}}$ \\
        \hdashline%
      && $z = \text{Diag}(v_1,\ldots,v_n %
                                  )(Wx)$\\
      && $z_{c_j} = z_{c_j}{-}y_j$ for $j=1..r$\\
      && $Az\checks{=}0$\\
 
        \hdashline

      &$\begin{array}{|c|}
        \hline
        \text{\scriptsize Protocol~\ref{cert:RPM:linearcomm} on $A_{\mathcal{I}, \mathcal{J}}$} \\
             \hdashline%
             \xrightarrow{\mathcal{R}_r = RPM(A_{\mathcal{I},\mathcal{J}}) } \\
             \hline
      \end{array}$&\\\hdashline

      & & $\mathcal{R}_A = \mathcal{E}_{m, \mathcal{I}}
       \mathcal{R}_r  \mathcal{E}_{n, \mathcal{J}}^{T} $\\\hline

  \end{tabular}
  \caption{Rank profile matrix}
  \label{cert:RPM:full}
  }
  }
\vspace{\spaceafterprotocol}\end{protocol}

\section{Certificate for the signature of an integer matrix}\label{ssec:sign}
The signature of a symmetric matrix is the triple $(n_+,n_-,n_0)$ indicating the
number of positive, negative, and zero eigenvalues, respectively.
Just like~\cite[Theorem~5]{dk14}, the idea is that the Prover commits the
signature, and then certifies it modulo a Verifier chosen prime.  
This works directly for the signature algorithm in~\cite[Corollary~1]{kns11}
together with the {\sc{CharPoly}} protocol of
\cite[Theorem~14]{jgd:2016:gammadet}. 
As in \S~\ref{sec:det}, in the dense case we propose here to replace
the {\sc{CharPoly}} computation with a symmetric Gaussian elimination. 

Over the rationals, an algorithm for the Prover could be to first compute and
certify the rank of $A$, and to compute a
permutation matrix $P$ such that
$P^TAP$ has generic rank profile: for instance compute a $PL_p{\Delta_p}L_p^TP^T$
factorization modulo a sufficiently large prime $p$.
Then $B=[I_r|0]P^TAP\begin{smatrix} I_r\\0\end{smatrix}$ is symmetric and non-singular.
It is then sufficient to lift or reconstruct only the block diagonal matrix
${\Delta}$ over $\Q$ of a non-pivoting symmetric factorization of $B$ (the unit
triangular matrix over $\Q$ need not be computed). Compared to an integer
characteristic polynomial computation this gains in practice an order of
magnitude in efficiency for the Prover as shown on the logscale
Figure~\ref{fig:signature}, using LinBox-1.5.1~\cite{linbox:2017}.  

\begin{figure}[htbp]
\centering%
\input{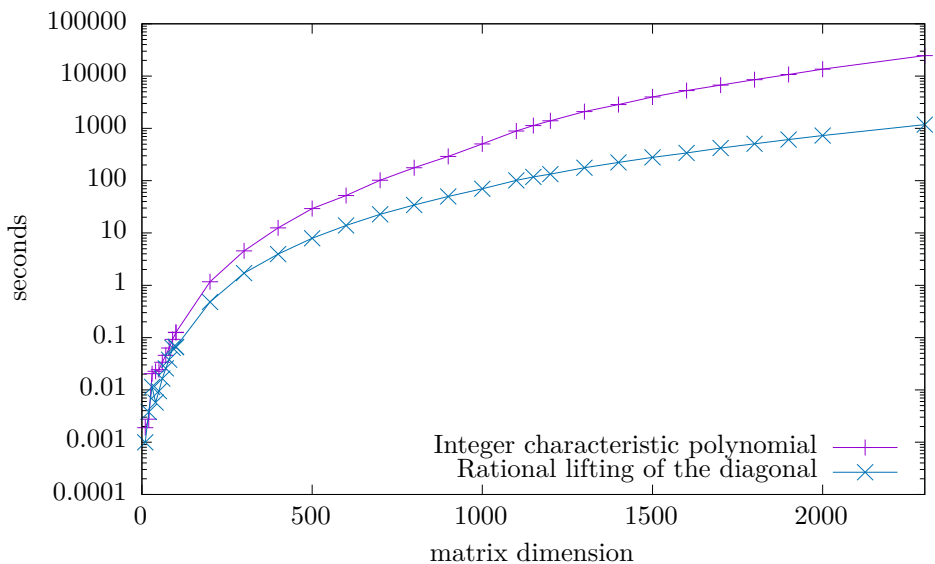}
\caption{
(Verifiable) signature computation on a single Intel Skylake core @3.4GHz.}\label{fig:signature}%
\end{figure}%

For the verification, the block diagonal matrix ${\Delta}$, and the permutation
$P$ are committed. The Verifier then randomly chooses a prime $q$ and enters an
interactive certification process for $P$ and $\Delta\mod{q}$ using
Protocol~\ref{cert:pldu}, as shown on Protocol~\ref{cert:sign}.  

\begin{protocol}[htbp]
  \centering
  {\small
    \begin{tabular}{|l c l|}
      \hline\rule{0pt}{12pt}   
      Prover & & \hspace{15pt}Verifier \\
      \multicolumn{3}{|c|}{\hspace{-25pt}$A \in \mathbb{\Z}^{n{\times}n}$ symmetric} \\
      \hdashline\rule{0pt}{12pt}    
       & $\fxrightarrow{\mathcal{I}}$ & \\
       & $\fxleftarrow{q_1}$ & $q_1\sample{S_1\subset\Primes}$\\
       & $\ldots$ & $\mathcal{I} \checks{=} RRP(A)=CRP(A) \mod q_1$, \\
       &  & and $|\mathcal{I}|\checks{=} rank(A) \mod q_1$ by Cert.~\ref{cert:crp}.\\ 
      \hdashline\rule{0pt}{12pt}   
       & $\fxrightarrow{P,\Delta}$ & $P\checks{\in}\mathcal{S}_r$, $\Delta\checks{\in}\mathcal{D}^{(2)}_r(\Q)$\\
      & $\fxleftarrow{q_2}$ & $q_2\sample{S_2\subset\Primes}$\\
      & $\begin{array}{|c|}
        \hline
        \text{\scriptsize Protocol~\ref{cert:pldu} on}~P^TA_{\mathcal{I}}P\mod{q_2}\\
        \hdashline%
        \xleftarrow{ \begin{smatrix} \widetilde{\phi_i}&\widetilde{\psi_i} \end{smatrix}
        } \\
        \xrightarrow{\begin{smatrix}\widetilde{x_i}&\widetilde{y_i}\end{smatrix}} \\
         \xleftarrow{ \begin{smatrix} \widetilde{\lambda_i} \end{smatrix}
        } \\
        \xrightarrow{\begin{smatrix}\widetilde{z_i}\end{smatrix}} \\
       \hline
      \end{array}$
      & \\
     && $\vdots$\\
    && \hspace{-15pt}$z^T{\Delta} \begin{bmatrix}x&y\end{bmatrix} \checks{=}
    (\lambda^T P^TA_{\mathcal{I}}) P \begin{bmatrix}\phi&\psi\end{bmatrix} \mod{q_2}$ \\
      \hdashline\rule{0pt}{12pt}   
      & & Extract ($n_+({\Delta}),n_-({\Delta}),n-r$)\\
      \hline
    \end{tabular}
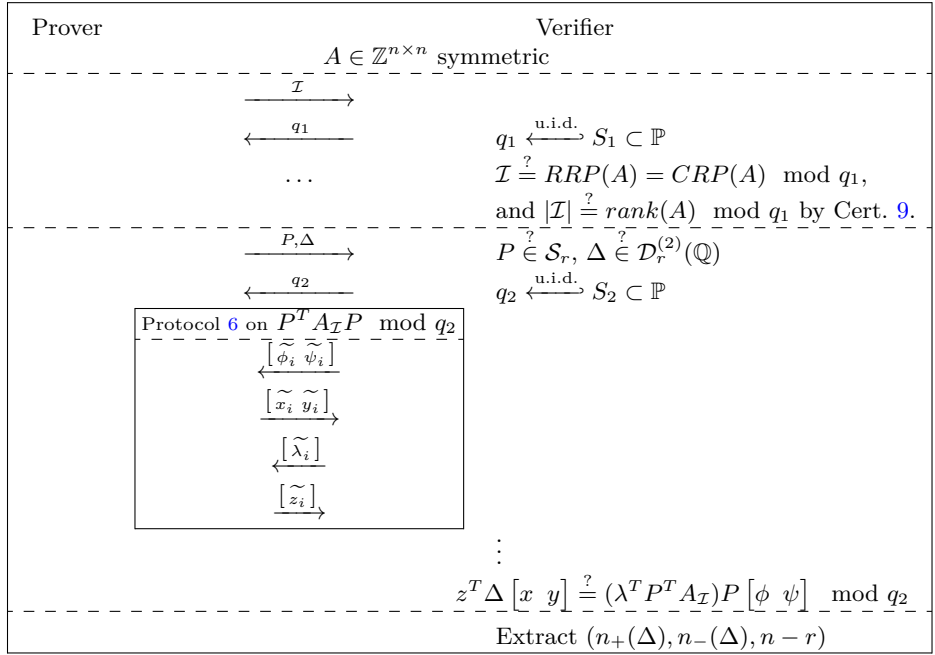
    \caption{Certificate for the signature of a symmetric matrix}
    \label{cert:sign}
  }
\end{protocol}

From~\cite[Theorem~5]{dk14}, we let $h = \log_2(\sqrt{n}^n ||A||_\infty^n)$ be the
logarithm of Hadamard's bound for the invariant factors of $A$. 
There cannot be more than $h$ primes reducing the rank. Therefore it is possible
to sample $c \cdot h$ distinct primes of magnitude bounded by $O(h \log(h))$ for
any constant $c>2$ and select $q_1$ from that set $S_1$. 
Once the rank is certified, the Prover can compute the permutation and lift the
diagonal. Finally the rational $PL{\Delta}L^TP^T$ factorization of the full rank
matrix can be similarly verified modulo a prime $q_2$. As for the determinant,
no more than $h$ primes can reduce the rank of ${\Delta}$ and $q_2$ can be
selected from the same kind of set. We have proven:
\begin{corollary} 
  For a symmetric matrix $A \in \mathbb{\Z}^{n{\times}n}$,
  certificate~\ref{cert:sign} for its signature is sound and perfectly complete.
\end{corollary}

The communication comprise that of the Certificate~\ref{cert:pldu}, the
permutation matrix $P$, all of size $n$, as well as small primes bounded by $h$,
and finally ${\Delta}$. 
Just like that of the characteristic polynomial, the size of ${\Delta}$ can be
quadratic and therefore the whole protocol is not linear. 
Thus a simpler quadratic certificate communicating the triangular matrix $L$
modulo $q_2$, and checking the decomposition $A-L{\Delta}L^T$ via Freivalds' 
certificate might also work. But then the communication and Verifier time would
always be quadratic.
Instead, Protocol~\ref{cert:sign}, just like the Protocol using the
characteristic polynomial, is better if the size of the
determinant is small, as then the size of ${\Delta}$ might be much less than
that of $L$ (for instance linear if the determinant is a constant). 
Protocol~\ref{cert:sign} is also interesting if $\mu(A)$ is less than
quadratic.

\section{Constant round certificates}\label{sec:mulmuley}
When delegating computations, the network latency can 
make communication rounds expensive. It can therefore also be interesting not
only to reduce the communication volume, but also the number of rounds. We
therefore propose in this section a certificate with a constant number of rounds
for triangular equivalence, still preserving Prover efficiency as well as linear
communication volume and Verifier cost. 
This applies then directly, as previously shown, to row or column rank
profiles. However it fails to apply to the generic rank profile, at least in a
straightforward manner, and we were unable to produce such  a certificate in
constant round for this task.

\subsection{Representative Laurent polynomial of a matrix}
Following a technique in~\cite{Mul86}, we first define the {\em representative
  Laurent polynomial}, $P_A(X)$ of an 
$m{\times}n$ matrix $A$ as~:
$$
P_A(X)=
\begin{bmatrix}1&X&X^2&\ldots&X^{m-1}\end{bmatrix}
\cdot{}A
\cdot\begin{bmatrix}1\\X^{-1}\\\vdots\\X^{1-n}
\end{bmatrix}
=\sum_{i=1}^{m}\sum_{j=1}^{n} A_{i,j}  X^{i-j}
$$

Therefore, if a matrix is lower triangular, then its representative
Laurent polynomial cannot have negative powers and it is therefore a
polynomial of degree at most $m-1$. 
The converse is not true, consider
for instance an upper diagonal with two opposite coefficients~:
$A_{i,i+1}=0$ for all $i$ except $A_{1,2}=-A_{2,3}$. Generically, if
one pre-multiplies $A$ on the right by a random non-zero diagonal
matrix, these cancellations will not occur as in general
$d_1A_{1,2}{\neq}-d_2A_{2,3}$ unless $A_{1,2}=A_{2,3}=0$.

\subsection{Constant round triangular equivalence certificate}
From this representation we can obtain a triangular equivalence
certificate that requires only a constant number of rounds: the Prover
commits that polynomial, then the Verifier will evaluate the
polynomial at a random point and compare this to the actual
projections. The counterpart is that the field size must be
sufficiently large so that the polynomial identity testing does not
fail. The full certificate is given in Protocol~\ref{cert:treqpol}. 
It requires that the Prover solves a regular system (this is checked
deterministically by reapplying the resulting vector), and a
preconditioning by a diagonal matrix to prevent cancellations.

\begin{protocol}[htbp]
  \centering
  \begin{tabular}{|c c l|}
    \hline\rule{0pt}{12pt}   
    Prover &  & Verifier \\
    & $ A, B \in \mathbb{F}^{m{\times}n} $ &  \\
    & $A$ is regular, $m{\geq}n$ & \\
    \hdashline\rule{0pt}{12pt}
    $\exists{}L$ lower triang. s.t. $AL= B$ & $\fxrightarrow{\exists{}L}$ & \\
    \hdashline\rule{0pt}{12pt}   
        & & Choose $S \subset \mathbb{F}$ \\
    & $\fxleftarrow{D}$& $D\sample{{\mathcal D}_n(S\backslash\{0\})}$ \\
    $g(X)=P_{LD}(X)$      & $\fxrightarrow{g(X)}$& $g\stackrel{?}{\in}
                            \F[X]_{\text{\upshape deg}\le n-1} %
                                                    $
\\
    & $\fxleftarrow{\lambda}$& $\lambda\sample{S}$ \\
    $y$, s.t. $A\cdot{}y=B\cdot{}D\cdot{}\begin{bmatrix}1\\\lambda^{-1}\\\vdots\\\lambda^{1-n}\end{bmatrix}$
    & $\fxrightarrow{y}$
    & $A\cdot{}y\checks{=}B\cdot{}D\cdot{}\begin{bmatrix}1\\\lambda^{-1}\\\vdots\\\lambda^{1-n}\end{bmatrix}$ \\
    && $g(\lambda)\checks{=}\begin{bmatrix}1& \lambda&\ldots &\lambda^{n-1}\end{bmatrix}\cdot{}y$ \\
    \hline
  \end{tabular}
  \caption{Constant round linear communication certificate for triangular equivalence}
  \label{cert:treqpol}
\end{protocol}

\begin{theorem}
  Let $A, B \in \mathbb{F}^{m \times n}$, $m{\geq}n$, and assume $A$ is regular.
  Certificate~\ref{cert:treqpol} is sound, with probability 
larger
  than $1-2\frac{n-1}{|S|}$ and perfectly complete.
  The Prover cost is dominated by one system solving, $O(mn^{\omega-1})$,
  the communication space is bounded by $3n+1$ 
  and the Verifier cost is bounded by $\mu(A)+\mu(B)+7n$
\end{theorem}
\begin{proof}
Let $x=D\begin{bmatrix}1\\\lambda^{-1}\\\vdots\\\lambda^{1-n}\end{bmatrix}$.
As $A$ is regular, there is only one solution $y$ to $Ay=Bx$, 
and $y=Lx$. Therefore 
$\begin{bmatrix}1& \lambda&\ldots &\lambda^{n-1}\end{bmatrix}\cdot{}y = 
\begin{bmatrix}1& \lambda&\ldots &\lambda^{n-1}\end{bmatrix}\cdot{}LD\begin{bmatrix}1\\\lambda^{-1}\\\vdots\\\lambda^{1-n}\end{bmatrix}=P_{LD}(\lambda)$ and the
protocol is correct.
For the soundness: 
\begin{itemize}
\item
  As $A$ is regular, there is only one solution $y$ to $Ay=Bx$, thus that check 
  ensures that $y$ is correct,
unless not all columns in $B$ are in the column space of $A$, which %
is handled as in the proof of Theorem~\ref{th:triangular_equivalence}. %
\item If $L$ is not triangular then its upper part is not identically
  zero. Therefore by considering $D$ as a diagonal matrix of
  indeterminates, at least one coefficient of negative degree of the
  representative rational fraction $LD$ will be non identically zero.
  As those are of degree $1$ in the indeterminates of $D$, for a
  random diagonal $D$, the representative rational fraction of $LD$
  will not be a polynomial with probability at least $1-\frac{1}{|S|-1}$.
\item If $g$ is not a polynomial of degree at most $n-1$, it is not the
  representative of a triangular matrix.
\item If $g$ is not the representative polynomial of
  $LD$ then by the DeMillo-Lipton/\allowbreak Schwartz/\allowbreak Zippel
  lemma~\cite{DeMillo:1978:ipl,Zippel:1979:ZSlemma,Schwartz:1979:SZlemma},
  its evaluation at $\lambda$ will fail with probability $1-2\frac{n-1}{|S|}$ (since
  $X^{n-1}(g-P_{LD})(X)$ is a polynomial of degree at most $2(n-1)$).
\end{itemize}
For the complexity, the Prover computes $L$, in $O(mn^{\omega-1})$. Then
$P_{LD}(X)$ requires one pass over the coefficients of $L$, and finally
$y=LD\begin{bmatrix}1\\\lambda^{-1}\\\vdots\\\lambda^{1-n}\end{bmatrix}$.
The communication cost is $D$, $g(X)$, $y$ all of size $n$, and $\lambda$.
The Verifier cost is, $\mu(A)+\mu(B)$ to apply $A$ and $B$, as well as $2n-3$ to
compute $\begin{bmatrix}1 & \lambda &\ldots & \lambda^{n-1}\end{bmatrix}$ and their 
inverses, $n-2$ to multiply by $D$, $2(n-1)$ to evaluate $g$, and
$2(n-1)$ to compute the dotproduct
$\begin{bmatrix}1 & \lambda &\ldots & \lambda^{n-1}\end{bmatrix} \cdot{}y$. 
\end{proof}

\subsection{Constant round certificates for the row and column rank
  profiles}\label{ssec:constcrp}
Now we can combine the lower rank Certificate~\ref{cert:lower_rank}, with the
constant-round Certificate~\ref{cert:treqpol} for triangular equivalence, as a
replacement of Certificate~\ref{cert:triangular_equivalence}, within
the column rank profile Certificate~\ref{cert:crp}, in order to get the
constant-round Certificate~\ref{cert:constcrp} for column rank profile. It
remains Prover efficient, linear in communication volume and Verifier time.

\begin{protocol}[htb]
  \centering
  {\def\arraystretch{1.2}
    \begin{tabular}{|l c l|}
        \hline
        Prover & & Verifier \\
        & $A \in \mathbb{F}^{m\times n}$ & \\
        \hdashline\rule{0pt}{12pt}
        $\mathcal{J}=(c_1,.., c_r)$ CRP of $A$ &
        $\fxrightarrow{\text{$\mathcal{J}$}}$ &
        Choose $S \subset \mathbb{F}$\\
        \hdashline%
        $\begin{array}{c}
          \\
          \\
          \beta~\text{s.t.}~A\beta=\nu\\
        \end{array}$
        &
        $\begin{array}{|c|}
          \hline
          \text{\scriptsize Protocol~\ref{cert:lower_rank}}\\
          \hdashline%
          \fxleftarrow{\nu} \\
          \fxrightarrow{\beta} \\
          \hline
        \end{array}$ 
        & 
        $\begin{array}{c}
          \alpha=\mathcal{E}_{m,\mathcal{J}}(\sample{S^r})\\
          \nu=A\alpha\\
          \beta\checks{=}\alpha\\
        \end{array}$\\
        \hdashline%
        $V=\text{Diag}(v_1,\ldots,v_n %
                      )W$
        (see (\ref{eq:Vdef})) %
        & $\fxleftarrow{v} $ & $v \sample{S^{n}}$\\
        $\Gamma$ upper tri.  s.t. $A_{*,\mathcal{J}}\Gamma = AV$ &
&    $W = [\mathds{1}_{i<c_{j+1}} ]\in\F^{n{\times}r}
        $\\

        \hdashline%
&$\begin{array}{|c|}
        \hline
          \text{\scriptsize Protocol~\ref{cert:treqpol}}\\
           \text{on $A_{*,\mathcal{J}}$
            and $B=AV$}\\
             \hdashline%
             \fxleftarrow{D} \\
             \fxrightarrow{g(X)} \\
             \fxleftarrow{\lambda} \\
             \fxrightarrow{y} \\
             \hline
      \end{array}$ & \\
        \hdashline%
 && $z = \text{Diag}(v_i)W D \begin{bmatrix}1\\\lambda^{-1}\\\vdots\\\lambda^{1-r}\end{bmatrix}$\\
 && $z_{c_j} = z_{c_j}{-}y_j, j=1..r$\\
 && $Az\checks{=}0$\\
 && $g(\lambda)\checks{=}\begin{bmatrix}1& \lambda&\ldots &\lambda^{r-1}\end{bmatrix}y$ \\
        \hline
    \end{tabular}
    \caption{Constant-round certificate for the column rank profile}
    \label{cert:constcrp}
    }
\vspace{\spaceafterprotocol}\end{protocol}

\begin{corollary}
For an $m{\times}n$ matrix of rank $r$, Certificate~\ref{cert:constcrp} is
sound and perfectly complete. 
It requires $3$ rounds, a volume of communication of $m+n+5r+1$
and less than $2\mu(A)+n+9r$ operations for the Verifier.
\end{corollary}

\section{Conclusion}
A summary of our contributions is given in
Table~\ref{contributions}, to be compared with the state of the art
in Table~\ref{SotA_matrix_rank}.
\begin{sidewaystable}[htbp]%
\centering
\begin{tabular}{%
l%
@{\hspace*{\colsp}}l%
c%
@{}c@{}%
l%
@{\hspace*{\colsp}}c%
@{\hspace*{\colsp}}l%
@{\hspace*{\colsp}}l %
}%
\toprule
& Algorithm
& Rounds
& \multicolumn{2}{c}{Prover}
& \multirow{2}{*}{\hspace*{\colsp}Communication}
& Probabilistic &
\multirow{2}{*}{$|S|$}
\\ %
\cmidrule{4-5}
&  &  & Determ. &
\multicolumn{1}{c}{Time}& & Verifier Time
\\
\midrule
\multirow{4}{*}{\sc{Rank}}
& \cite{kns11} over \cite{ckl13} & No & No
& $\SO{r^{\omega} \swplus  \mu(A)}$
& $\SO{r^2\swplus m\swplus n}$
& $\SO{r^{2} \swplus  \mu(A)}$&$\geq 2$
\\
&\multirow{2}{*}{\cite{dk14}} & \multirow{2}{*}{$2$} & \multirow{2}{*}{No}
& $O(n(\mu(A)\swplus n))$ %
&\multirow{2}{*}{$O(m\swplus n)$}
&\multirow{2}{*}{$2\mu(A)\swplus \SO{m\swplus n}$}
&\multirow{2}{*}{$\Omega(\min\{m,n\} \log(mn))$%
}
\\
& & &  & or $O(mnr^{\omega -2})$ &  & &
\\
&\cite{eberly15} & $2$ & Yes
& $O(mnr^{\omega -2})$
& $O(n\swplus r)$
& $O(\mu(A) \swplus   n)$ & $\geq 2$
\\
\midrule
CRP/
& \cite{kns11} over \cite{sy15} & No & No
& $\SO{r^{\omega} \swplus m\swplus n\swplus  \mu(A)}$
& $\SO{r^2\swplus m\swplus n}$
& $\SO{r^2\swplus  m \swplus  n \swplus \mu(A)}$
& $\Omega(\min\{m,n\} \log(mn))$%
\\
\quad RRP
& \cite{kns11} over \cite{jps13} & No & Yes
& $O(mnr^{\omega-2})$ & $\SO{mn}$ & $\SO{mn}$ &$\geq 2$
\\
\midrule
\multirow{2}{*}{RPM} & \cite{kns11} over \cite{jgd:2017:bruhat} & No & No
& $\SO{r^{\omega} \swplus  m \swplus  n \swplus  \mu(A)}$
& $\SO{r^2\swplus m\swplus n}$
& $\SO{r^2 \swplus  m \swplus  n \swplus  \mu(A)}$
& $\Omega(\min\{m,n\}\log(mn))$%
\\
& \cite{kns11} over \cite{DPS:2013} & No & Yes
& $O(mnr^{\omega-2})$
& $\SO{mn}$ & $\SO{mn}$ & $\geq 2$
\\
\midrule
\multirow{3}{*}{\textsc{Det}} & \cite{freivalds79} \& PLUQ & No & Yes
& $O(n^\omega)$ & $O(n^2)$
& $O(n^2) \swplus  \mu(A)$ & $\geq 2$
\\
& \multirow{2}{*}{\cite{jgd:2016:gammadet} \& {\textsc{CharPoly}}}
& \multirow{2}{*}{$2$} & \multirow{2}{*}{No}
& $O(n\mu(A))$ %
& \multirow{2}{*}{$O(n)$} & \multirow{2}{*}{$\mu(A)\swplus O(n)$} 
& \multirow{2}{*}{$\geq n^2$}
\\
& & & & or $O(n^\omega)$ & & &\\
\bottomrule
\end{tabular}
\caption{%
\vtop{\hbox{%
State of the art certificates for the rank, the row and column rank
profiles, the rank profile matrix}
\hbox{%
and the determinant}
\hbox{\ }%
}%
}%
\label{SotA_matrix_rank}
\centering
\addtolength{\colsp}{0.3em}
\begin{tabular}{%
l%
@{\hspace*{\colsp}}l%
@{\hspace*{\colsp}}c%
@{}c@{}%
l%
@{\hspace*{\colsp}}c%
@{\hspace*{\colsp}}l%
@{\hspace*{\colsp}}l %
}%
\toprule
& Algorithm
& Rounds
& \multicolumn{2}{c}{Prover}
& \multirow{2}{*}{\hspace*{\colsp}Communication}
& Probabilistic &
\multirow{2}{*}{$|S|$}
\\ %
\cmidrule{4-5}
&  &  & Determ. &
\multicolumn{1}{c}{Time}& & Verifier Time
\\
\midrule
\multirow{3}{*}{CRP/RRP}
& \S~\ref{sec:noninterractive:CRP}& No & Yes
& $O(mnr^{\omega-2})$
& $O(r(m\swplus n))$
& $O(r(m\swplus n))\swplus  \mu(A)$ & $\geq 2$
\\
& \S~\ref{sec:interractive:CRP}
& $O(n)$  & Yes & $O(mnr^{\omega-2})$
& $O(m\swplus n)$
& $2\mu(A)\swplus O(m\swplus n)$ & $\geq 2$
\\
& \S~\ref{ssec:constcrp}
& $3$ & Yes
& $O(mnr^{\omega-2})$
& $O(m\swplus n)$
& $2\mu(A)\swplus O(m\swplus n)$ & $\geq 2n-1$
\\
\midrule
\multirow{2}{*}{RPM} & \S~\ref{sec:noninterractive:RPM} & No & Yes
&$O(mnr^{\omega-2})$  & $O(r(m\swplus n))$ & $O(r(m\swplus n)) \swplus  \mu(A)$ &$\geq 2$ \\
& \S~\ref{sec:interractive:RPM}
&  $O(n)$ & Yes
& $O(mnr^{\omega-2})$
& $O(m\swplus n)$
& $4\mu(A)\swplus O(m\swplus n)$
& $\Omega(n)$ %
\\
\midrule
{\textsc{Det}} & \S~\ref{sec:det} \& PLUQ
& $O(n)$ & Yes
& $O(n^\omega)$
& $O(n)$
& $\mu(A)\swplus O(n)$
& $\Omega(n)$ %
\\
\bottomrule
\end{tabular}
\caption{This paper's contributions}
\label{contributions}
\end{sidewaystable}

We have provided certificates that can save overall computational time for the
Provers and an order of magnitude in terms of communication volume or number of
rounds. 
Table~\ref{tab:io} compares linear and quadratic communications, as
well as sub-cubic (PLUQ, {\sc{CharPoly}}) or quadratic matrix operations. 
These results show first that it is interesting to use linear space certificates
even when they have quadratic Verification time. 
The table also presents a practical constant
factor of about 5 between PLUQ and {\sc{CharPoly}} computations. 

One key idea in our contribution is to certify the existence of a triangular matrix in an
equivalence relation, by having an $n$ round protocol where data dependency
matches the triangular shape of the unknown matrix factor.
This approach was successfully adapted to the certificate of generic rank
profileness, where now two triangular unknown 
triangular factors are considered,
in the LU decomposition.

Mulmuley's Laurent's polynomial representation of a matrix successfully replaces
the former technique to certify triangular equivalence, and consequently row or
column rank profiles, reducing the number of rounds from linear to constant.
However, we were unable to adapt this technique for the certificate for generic
rank profileness, and consequently for certifying a rank profile matrix.

The use of symmetric Gaussian elimination allowed us to achieve a more practical
certificate for the signature of symmetric integer matrices. Even though it is
based on LDLT certificates with linear communication modulo a prime, the
diagonal of rational eigenvalues remains quadratic in size, and full precision
was required to recover their sign.
Designing a linear communication, Prover efficient protocol to certify the
signature is the other major 
open problem which should be investigated in the continuation of this work.

%
%
%
%
%
%
%
%

\providecommand{\noopsort}[1]{}

\end{document}